\documentclass[aps,prx,amsmath,amssymb,twocolumn,10pt,superscriptaddress, a4paper]{revtex4-2}

\usepackage{lmodern}
\usepackage[T1]{fontenc}
\usepackage[utf8]{inputenc}

\usepackage{enumitem}

\usepackage{changes}

\usepackage{mathrsfs}

\usepackage{amsfonts}
\usepackage{amsmath,amssymb,amsthm}
\usepackage{graphicx}
\usepackage{bbm}
\usepackage{hyperref}
\usepackage{mathtools}
\usepackage{physics}
\usepackage{upgreek}
\usepackage{dsfont}
\usepackage[capitalize]{cleveref}
\usepackage{xcolor}

\usepackage{geometry}
\newtheorem{thm}{Theorem}
\newtheorem{coro}{Corollary}
\newtheorem*{thm*}{Theorem}
\makeatletter
\newcommand{\setthmtag}[1]{
  \let\oldthethm\thethm
  \newcommand{\thethm}{#1}
  \g@addto@macro\endthm{
    \addtocounter{thm}{-1} 
    \global\let\thethm\oldthethm}
  }
\makeatother

\newtheorem*{prop*}{Proposition}
\newtheorem{lemma}[thm]{Lemma}
\newtheorem*{lemma*}{Lemma}

\newtheorem*{cor*}{Corollary}

\newtheorem*{cj*}{Conjecture}
\newtheorem{Def}[thm]{Definition}
\newtheorem*{Def*}{Definition}

\theoremstyle{definition}

\newtheorem*{rem*}{Remark}

\def\beq{\begin{equation}}
\def\eeq{\end{equation}}
\def\bq{\begin{quote}}
\def\eq{\end{quote}}
\def\ben{\begin{enumerate}}
\def\een{\end{enumerate}}
\def\bit{\begin{itemize}}
\def\eit{\end{itemize}}

\def\r|{\right|}

\newcommand{\n}{\mathcal{N}}

\newcommand\be{\begin{equation}}
\newcommand\ee{\end{equation}}

\newcommand{\pdftitle}{High-temperature partition functions and classical simulatability of long-range quantum systems}

\hypersetup{
 colorlinks,
 linkcolor={red!50!black},
 citecolor={blue!60!black},
 urlcolor={blue!80!black},
 unicode,
 pdftitle={\pdftitle}, 
 pdfauthor={Jorge S\'anchez-Segovia, Jan T. Schneider, \'Alvaro M. Alhambra},
 pdfduplex = DuplexFlipLongEdge,
 pdflang = en,
}

\usepackage{subcaption}
\usepackage{caption}
\usepackage{ragged2e}
\justifying

\usepackage[
 activate={true,nocompatibility},
 tracking=true,
 kerning=true,
]{microtype}
\microtypecontext{spacing=nonfrench}

\begin{document}

\title{\pdftitle}

\author{\begingroup
\hypersetup{urlcolor=black}
\href{https://orcid.org/0009-0002-1274-4747}{Jorge~Sánchez-Segovia
\endgroup}
}
\email{jorge.sanchezsegovia@estudiante.uam.es}
\affiliation{Instituto de Física Teórica UAM/CSIC, C. Nicolás Cabrera 13-15, Cantoblanco, 28049 Madrid, Spain}
\affiliation{Institute of Fundamental Physics IFF-CSIC, C. Serrano 113b, Madrid 28006, Spain}

\author{\begingroup
\hypersetup{urlcolor=black}
\href{https://orcid.org/0000-0001-9854-998X}{Jan~T.~Schneider
\endgroup}
}
\email{jan.schneider@iff.csic.es}
 \affiliation{Institute of Fundamental Physics IFF-CSIC, C. Serrano 113b, Madrid 28006, Spain}

\author{\begingroup
\hypersetup{urlcolor=black}
\href{https://orcid.org/0000-0002-5889-4022}{Álvaro~M.~Alhambra
\endgroup}
}
\email{alvaro.alhambra@csic.es}
 \affiliation{Instituto de Física Teórica UAM/CSIC, C. Nicolás Cabrera 13-15, Cantoblanco, 28049 Madrid, Spain}

\begin{abstract}
\setlength{\parindent}{0pt}

Long-range quantum systems, in which the interactions decay as $1/r^{\alpha}$, are of increasing interest due to the variety of experimental set-ups in which they naturally appear. Motivated by this, we study fundamental properties of long-range spin systems in thermal equilibrium, focusing on the weak regime of $
\alpha>D$. Our main result is a proof of analiticity of their partition functions at high temperatures, which allows us to construct a classical algorithm with sub-exponential runtime $\exp(\mathcal{O}(\log^2(N/\epsilon)))$ that approximates the log-partition function to small additive error $\epsilon$. As by-products, we establish the equivalence of ensembles and the Gaussianity of the density of states, which we verify numerically in both the weak and strong long-range regimes. This also yields constraints on the appearance of various classes of phase transitions, including thermal, dynamical and excited-state ones. Our main technical contribution is the extension to the quantum long-range regime of the convergence criterion for cluster expansions of Koteck\'y and Preiss.

\end{abstract}

\maketitle

\section{Introduction}
Aided by the fast-paced development of experimental quantum platforms, we are increasingly able to probe complex systems of interacting quantum particles. In many of these platforms, the interactions between the individual particles are naturally \emph{long ranged}, with all-to-all couplings decaying algebraically with the distance $r$ between two constituents as $1/r^\alpha$ with \(0 < \alpha < \infty\). This includes trapped ions~\cite{ReviewTrappedIons}, Rydberg atoms~\cite{Browaeys2020} and cold atoms~\cite{Douglas_2015,Hung_2016} or molecules~\cite{Cornish2024}.
The greater range of interactions than that of local spin models (e.g., nearest-neightbor Ising) allows for interesting physical phenomena~\cite{IslamExperiment,Richerme2014,Joshi_2022,JoshiPRL}, including inequivalence of ensembles~\cite{Barr__2001,Kastner_2010,Russomanno_2021,Defenu_2024}, negative heat capacity~\cite{campa2014e}, or metastability and discreteness of the spectrum~\cite{defenu2021c}. It can also have a dramatic impact on quantum effects such as spin squeezing~\cite{Bornet_2023,Eckner_2023}.

It is thus of wide interest to study general properties of long-range quantum systems.  
Extensive theoretical research has been recently conducted toward characterizing both the dynamics (e.g.~\cite{KastnerEquilibration,Hauke_2013, jan_PhysRevResearch.3.L012022,LRB_1PhysRevX.10.031010,LRB2_PhysRevLett.127.160401,Pappalardi_2018,Mori_2019,DPTsLongRange,DefenuReview_2024,mattes2025longrangeinteractingsystemslocally}) and equilibrium physics (e.g.~\cite{Barr__2001,Kastner_2010,Russomanno_2021,Defenu_2024,Gong_2017,Kuwahara_2020_AreaLaw,jan_PhysRevB.106.014306}).
Across the literature, two regimes are typically identified as fundamentally different~\cite{LRreview_Defenu_2023}: the strong ($\alpha \le D$) and the weak ($\alpha>D$) long-range regimes, with $D$ being the system's dimension.
While the strong regime is typically accompanied by some of the aforementioned anomalous phenomena, the weak regime is often closer to short-range systems.

Here, we study the thermal equilibrium properties of spin systems in the weak long-range regime. Our main result is that, at high enough temperatures, partition functions (and simple generalizations thereof) are analytic and well approximated by their Taylor series. We also show that this yields a subexponential time classical algorithm for their estimation to multiplicative error. This runtime contrasts with the exponential runtimes expected for arbitrary temperatures~\cite{Poulin_2009,Bravyi_2022} (due to QMA hardness), and the polynomial ones established for short-range models at high temperature~\cite{AlgorithmicPirogov–Sinai,efficient_alg_Mann_2021,Harrow_2020,Haah2024}.  

We also explore further consequences of that analyticity. In particular, we establish the ensemble equivalence at high temperatures (as previously reported in~\cite{Kuwahara_2020_Gaussian}), and the Gaussianity of their density of states (DOS)~\cite{Hartmann_2004,Rai_2024}. We also run numerical tests of the DOS with tensor network methods that go beyond our analytical results, and hint at deviations of that Gaussianity in the strong long-range regime. We finally explore  how our results constrain the appearance of different kinds of phase transitions in the strong long-range regime, including thermal, dynamical, and excited-state phase transitions~\cite{ESQPTs_Cejnar_2021}.

The primary analytical toolkit that we use is the \textit{cluster expansion}, a versatile technique rooted in mathematical physics, which offers a systematic way to describe expansions around reference points (in this case, $\beta=0$).
Our main technical contribution is a proof of the convergence of this expansion in weak long-range systems, by extending the Koteck\'y-Preiss~\cite{kotecky1986cluster} convergence criterion to that regime. This criterion has been previously applied to finite-range quantum systems, at both high~\cite{efficient_alg_Mann_2021} and low~\cite{Helmuth_2023} temperatures, and to more general classes of quantum problems~\cite{classsimofshort_PRXQuantum.4.020340,Mann_2024}.  Our technical findings are also inspired by recent results on cluster expansions for high-temperature finite-range quantum models~\cite{Kliesch_2014,Molnar_2015,yin2023polynomialtimeclassicalsamplinghightemperature,Haah2024,bluhm2024strongdecaycorrelationsgibbs,bakshi2025hightemperaturegibbsstatesunentangled}, some of which have already been extended to the weak long-range regime~\cite{Kuwahara_2020_Gaussian,kimura2024clusteringtheorem1dlongrange,kim2024thermalarealawlongrange}. 

We start with the setup and an introduction to the cluster expansion in Sec.~\ref{sec:setup}. We then show our main result for partition functions in Sec.~\ref{sec:main}, followed by the classical algorithm in Sec.~\ref{sec5_aproximationalgorithm}. The consequences for the statistical properties are shown in Sec.~\ref{sec:statistical} and the constraints on criticality in Sec.~\ref{sec3_absenceofcriticality}. All the proofs are outlined in the main text, with some of the more technical details placed in Appendixes A-C.

\section{Setup}\label{sec:setup}

We consider quantum systems with $N$ qudits on a $D$-dimensional lattice $\Lambda$ with $\vert \Lambda \vert =N$, governed by Hamiltonians with $k$-body terms
\begin{equation}
    H= \sum_{Z: \vert Z \vert \leq k} h_{Z}, \label{eq:k-body-term}
\end{equation}
with long-range interactions decaying as a polynomial, such that
\begin{equation}
    \sum_{Z \ni \{i,i'\}} \norm{h_Z} \leq \frac{g}{\left(1+ d_{ii'}\right)^\alpha} \equiv J_{ii'},
    \label{eq:HkLocal}
\end{equation}
where $g$ is some finite constant. Here the symbol $d_{ij}$ stands for the Manhattan distance between
the sites $i$ and $j$ in $\Lambda$. By taking $\alpha \rightarrow \infty$ we also include finite-range interactions.
A geometric quantity which will play an important role IN our calculations is
\begin{equation}
    u= 2^{\alpha}\max_i \sum_j \frac{1}{\left(1+d_{ij}\right)^\alpha}, \label{eq:def-u}
\end{equation}
which is $\mathcal{O}(1)$ if $\alpha > D$.

We focus on functions of $H$ of the form 
\begin{equation}
    Z_\rho = \text{Tr}[\mathrm{e}^{-\beta H} \rho],
\end{equation}
where $\beta \in \mathbb{C}$ and $\rho$ is a product state over all sites. By choosing $\rho=\mathbb{I}/d^N$ and $\beta >0$ we recover the usual partition function up to a factor of $1/d^N$, and by choosing $\beta = \mathrm{i} t $ with $t \in \mathbb{R}$, and $\rho= \otimes_i^N \ketbra{\Phi_i}{\Phi_i}$, we obtain a Loschmidt echo.

\subsection{Polymer models}

We now introduce a convenient way of writing the function $Z_ \rho$, in terms of so-called \emph{polymers}.

In an abstract setting, a polymer model is a triple $(\mathcal{C}, w, \sim)$, where $\mathcal{C}$ is a countable set whose objects are called polymers. The function $w : \mathcal{C} \to \mathbb{C}$ assigns to each polymer $\gamma \in \mathcal{C}$ a complex number $w_\gamma$ called the \emph{weight} of the polymer, and $\sim$ is a symmetric compatibility relation such that each polymer is incompatible with itself. Equivalently, the incompatibility relation $\nsim$ is a symmetric and reflexive relation. A set of polymers is called \textit{admissible} if all the polymers in the set are pairwise compatible. Note that the empty set is admissible.

Let $\mathcal{G}$ denote the collection of all admissible sets of polymers from $\mathcal{C}$. The polymer expansion of the partition function is defined by
\begin{equation} \label{eq:polymerexp}
Z_\rho = \sum_{\Gamma \in \mathcal{G}} \prod_{\gamma \in\Gamma} w_\gamma,
\end{equation}
where \( \Gamma \) is a nonempty ordered polymer tuple.

Here, the set of polymers $\mathcal{C}$ are identified with tuples of hyperedges in $\Lambda$, with each hyperedge corresponding to a particular term $h_Z$. We define $\norm{\gamma}$ as the total number of hyperedges in the polymer, with multiplicities $m_\gamma(Z)$ such that $\norm{\gamma}= \sum_Z m_{\gamma}(Z)$, and $\vert \gamma \vert$ the number of different hyperedges. The relation $\sim$ is such that $\gamma_1 \sim \gamma_2$ if the corresponding sets of hyperedges are disconnected and $\nsim$ otherwise. 
Finally, the weight $w_\gamma$ is given by (see App.~A of~\cite{efficient_alg_Mann_2021} or App.~\ref{app:KP_GenPartFunct}) 
\begin{equation}
    w_\gamma= \frac{ \left(-\beta\right)^{\norm{\gamma}}}{\norm{\gamma}! \prod_{Z \in \gamma}m_{\gamma}(Z)!} \text{Tr} \left[\sum_{\sigma \in S_{\norm{\gamma}}}  \prod_{i=1}^{\norm{\gamma}} h_{\gamma_{\sigma(i)}} \rho \right], 
    \label{eq::weight}
\end{equation}
where inside the trace we sum over permutations of $\norm{\gamma}$ elements. The factor $\prod_{Z \in \gamma} m_{\gamma}(Z)! $ is there to avoid overcounting when repeated hyperedges appear in a polymer. 

\subsection{The cluster expansion}

Consider a set of incompatible polymers $\Gamma$, where the underlying hypergraph of a polymer is connected to at least that of another polymer. We refer to these sets as  \emph{connected clusters}. The \emph{incompatibility graph} \( H_{\Gamma} \) of \( \Gamma \) is the graph with vertex set the polymers $ \{ \gamma \in \Gamma \} $ and edges between two polymers if and only if they are incompatible (see Fig.~\ref{fig:S2.connectedcanddisconnectedclusters}). By assumption, $H_\Gamma$ is connected.

The cluster expansion is then a formal power series for $ \log Z_\rho $ in the variables $ w_\gamma $, defined by
\begin{equation}\label{eq:logZexp}
    \log Z_{\rho} := \sum_{\Gamma \in  \mathcal{G}_C}  \varphi \left(H_{\Gamma} \right)  \prod_{\gamma \in \Gamma} w_{\gamma},
\end{equation}
where \( \mathcal{G}_C \) denotes the set of all connected clusters $\Gamma$ of polymers \( \gamma \in \mathcal{C} \). See e.g.~\cite{StatisticalMechanicsofLatticeSystems_friedli_velenik_2017} for a proof of Eq.~\eqref{eq:logZexp} and of why only connected clusters contribute, as opposed to admissible sets as in Eq.~\eqref{eq:polymerexp}.
Given a graph $H_\Gamma$, the \emph{Ursell function} is defined as 
\begin{equation}
 \varphi(H_\Gamma) := \frac{1}{\vert H_\Gamma\vert!} \sum_{\substack{S \subseteq E(H_\Gamma)\\ S \text{ spanning, connected}}} (-1)^{\vert S\vert},
\end{equation}
where $S$ are subsets of edges of $H_\Gamma$. This cluster expansion can also be seen as a different rearrangement of the Taylor expansion of a parameter~\cite{Dobrushin1996}, in this case $\beta$ around $\beta=0$.
\begin{figure}[t]
\centering
\renewcommand{\thesubfigure}{\alph{subfigure}}
    \begin{subfigure}[b]{0.4\linewidth}
    \centering
    \includegraphics[width=3.1 cm]{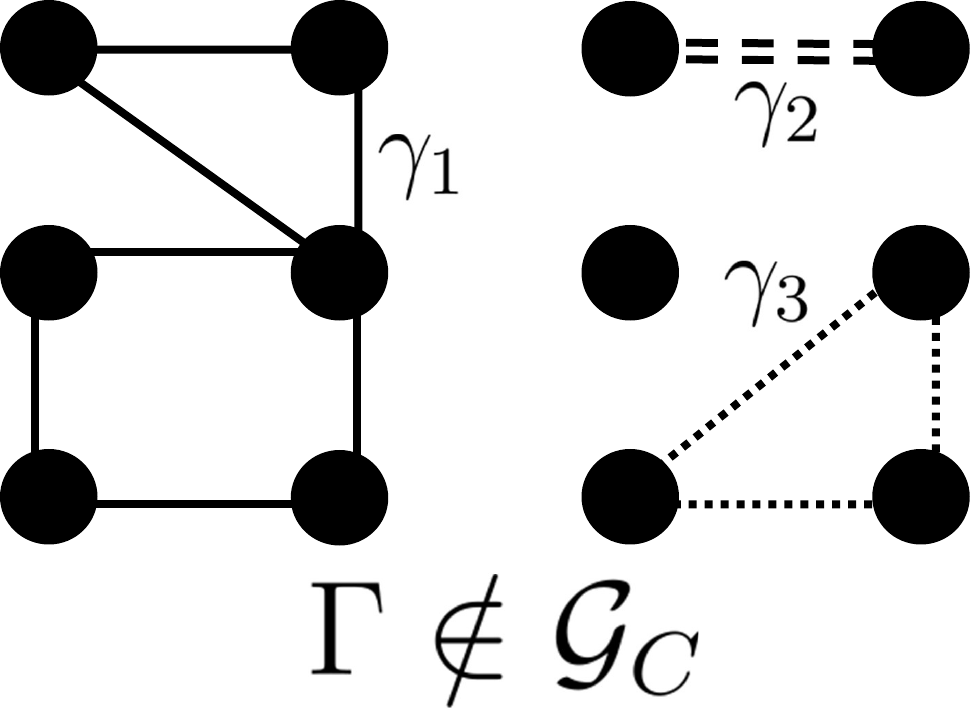}
   \caption{}
    \end{subfigure}
\hfill
    \begin{subfigure}[b]{0.4\linewidth}
    \centering
    \includegraphics[width=3.1 cm]{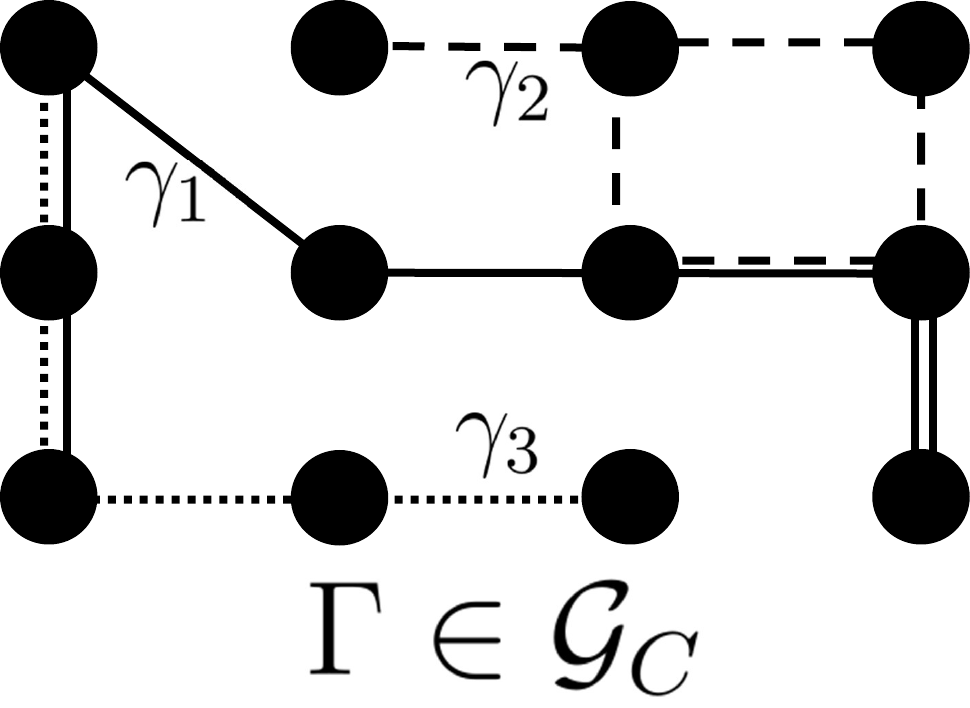}
    \caption{ }
    \end{subfigure}
   \caption{ \justifying Illustration of a set of polymers $\Gamma=\{\gamma_1,\gamma_2,\gamma_3\}$. In \textbf{(a)} the set is admissible, as they are disconnected, while in \textbf{(b)} the set is a cluster as they are connected.}
   \label{fig:S2.connectedcanddisconnectedclusters}
\end{figure}
For each set of polymers $\Gamma$, let us define  $\norm{\Gamma}=\sum_{\gamma \in \Gamma} \norm{\gamma}$.
A sufficient condition for convergence of the expansion of $\log Z_\rho$ is the following.
\begin{lemma}(Koteck\'y--Preiss.~\cite{kotecky1986cluster})
    Assuming there exist a function $a : \Gamma \rightarrow \mathbb{R}_{>0}$ such that, for every $\gamma^*\in \Gamma$
    \begin{equation}
        \sum_{\gamma \nsim \gamma^*} |w \left(\gamma\right)| \mathrm{e}^{a\left(\gamma\right)} \leq a \left(\gamma^*\right),
        \label{eq::Kotecky-Preiss}
    \end{equation}
then, 
    \begin{equation}
    \sum_{\substack{\Gamma \in \mathcal{G}_C\\ \Gamma \nsim \gamma^*}} \vert \varphi(\Gamma) \prod_{\gamma \in \Gamma}  w_{\gamma}\vert  \leq a \left(\gamma^*\right),
    \label{eq::KoteckyPreiss2}
    \end{equation}
and consequently $\log Z_\rho$ is analytic and has a convergent cluster expansion.
    \label{le::Kotecky-Preiss}
\end{lemma}
The proof dates back to the original work~\cite{kotecky1986cluster} and can also be found in~\cite{StatisticalMechanicsofLatticeSystems_friedli_velenik_2017}, with refinements of this criterion later proven in ~\cite{Dobrushin1996,Ueltschi2004,Fernandez_2007}.

\section{Main result}\label{sec:main}

Our main technical result is to show that the convergence criterion of Lemma~\ref{le::Kotecky-Preiss} holds for long-range systems at small enough $\beta$. More specifically, we prove that Eq.~\eqref{eq::KoteckyPreiss2} holds and that this implies a fast convergence of the Taylor series in $\beta$, which is defined by considering only clusters up to a given size. 
\begin{thm}
Let $H$ be a long-range Hamiltonian such that $\alpha> D$. Then, $\log Z_\rho$ is analytic and has a convergent cluster expansion for $\vert \beta \vert < \beta^* =\left( 8 \mathrm{e}k gu \right)^{-1}$, as
\begin{equation}\label{eq:mainresult}
    \left \vert \log Z_\rho -T_m \right \vert  \leq N \left\vert \frac{\beta}{\beta^*}\right\vert^m, 
\end{equation}
where the $m$th order expansion $T_m$ is
    \begin{equation}
        T_m := \sum_{\substack{\Gamma \in  \mathcal{G}_C\\ \norm{\Gamma} \leq m}}  \varphi \left(\Gamma \right)  \prod_{\gamma \in \Gamma} w_{\gamma}.
    \end{equation}
    \label{Theorem:KoteckyLR}
\end{thm}

\begin{proof}
    The full proof can be found in Appendix~\ref{app:KP_LR} but we sketch it here. We first choose $a\left(\gamma^*\right) = \abs{\gamma^*}$ in the LHS of Eq.~\ref{eq::Kotecky-Preiss} and fix a hyperedge $Z^*$. The sum over polymers connected to $Z^*$ can be seen as a sum over hypergraphs in the lattice. Since the interaction is long range, the resulting sum over polymers includes hypergraphs connected to $Z^*$ with all combinations of hyperedges. Each of these combinations, however, can be mapped in a nonunique way to a tree rooted at $Z^*$, in which a vertex of the tree corresponds to a lattice point. For each rooted tree, we upper bound the sum over all lattice points, which can be done in a controlled way due to the decay of the interactions with $\alpha >D$, in such a way that the upper bound is independent of the tree. We then just need to upper bound the number of trees with $m$ edges by $4^m$. Finally, after some manipulation of Eq.~\eqref{eq::KoteckyPreiss2}, we arrive to an upper bound on $\abs{\log Z_\rho - T_m }$.
\end{proof}
The proof requires a rather different argument to that of finite-range Hamiltonians~\cite{efficient_alg_Mann_2021,Helmuth_2023}, since the combinatorial arguments previously used to control the sum over polymers in Eq.~\eqref{eq::Kotecky-Preiss} do not apply for long-range models, where the interaction graph is complete. Instead, we have to resort to a different counting of the types of structures that the polymers can have, and then sum over the entire system aided by the \emph{uniform summability} and \emph{convolution} conditions of the long range interactions with $\alpha>D$ as defined in e.g.~\cite{Nachtergaele_2006} (this was named \emph{reproducing} in~\cite{hastings2010quasiadiabatic}). 

It is useful for some applications to generalize Theorem~\ref{Theorem:KoteckyLR} to partition functions with several exponentials within the trace.
 Given $K$ Hamiltonians $\{H_l \}_{l=1}^K$ and parameters $\lambda_l$, the generalized partition function is defined as
\begin{equation}
 \log Z_{Gen, \rho}=\log \text{Tr} \left[\prod_l \mathrm{e}^{\lambda_l H_l} \rho \right].
\end{equation}
\begin{thm}
If $ \sum_l |\lambda_l|\leq \beta^*/K$, $ \log Z_{Gen, \rho} $ has a convergent cluster expansion and
    \begin{equation}
       \left \vert \log Z_{Gen, \rho} -T_{M_1,  \dots M_K} \right \vert  \leq N \left(K \sum_{l} \bigg \vert \frac{\lambda_l}{\beta^*} \bigg \vert \right)^M,
    \end{equation}
    where $T_{M_1,  \dots M_K}$ is the expansion in $\{ \lambda_l \}$ to  order $\{M_1,  \dots M_K\}$ and $M= \sum_l M_l$.
    \label{theo:KP_GenPartFunct}
\end{thm}
The proof is left to Appendix~\ref{app:KP_GenPartFunct}. The strategy is very similar to the one given in Theorem~\ref{Theorem:KoteckyLR} but now with the extra degree of freedom in the sum over polymers, due to the $K$ types of hyperedges appearing in the expansion. 
\section{Approximation algorithm for partition functions}
\label{sec5_aproximationalgorithm}

The convergence of the Taylor series allows for a classical approximation algorithm of $\log Z_\rho$ with bounded runtime, as is already established for partition functions of finite-range models~\cite{AlgorithmicPirogov–Sinai,efficient_alg_Mann_2021,Harrow_2020,Haah2024}. The idea is to explicitly show how to compute the individual Taylor moments up to a given order. Then, their sum yields a good approximation due to Eq.~\eqref{eq:mainresult}. 
We first need a lemma to count all the connected clusters in the sum of Eq.~\eqref{eq:logZexp}.

\begin{lemma}
    The connected clusters of size at most $m$ can be listed in time $N^{\mathcal{O}(km)} m!$
    \label{lemma:alg_countingclusters}
\end{lemma}

\begin{proof}
    The first point is to enumerate all connected hypergraphs in $\Lambda$ of size at most $m$, including multiplicities. This can be done by generating them iteratively. We start by  fixing a specific vertex, of which there are $N$ possibilities. Then, there there are at most $(N-1)^k$ possible hyperedges to choose from to continue growing the hypergraph. After that, there are two sites, and at most $(N-2)^k$ edges to choose from. If we continue this process up to hypergraph size $m$, there are $\binom{N}{m}^k \left(m!\right)^2 \leq N^{km} m!$ hypergraphs $S_m$. Then, for each hypergraph $S_m$, we enumerate all polymers (multisets) of size at most $m$ whose corresponding hypergraph is $S_m$, and from that list of polymers, each connected cluster can be listed. The last two steps can be done in time $\exp{\mathcal{O}(m)}$, as shown in Theorem 6 of~\cite{AlgorithmicPirogov–Sinai}. Multiplying the two we obtain the runtime estimate.
\end{proof}

Lemma \ref{lemma:alg_countingclusters} is the one that differs more significantly with respect to the finite-range case, and constitutes the costliest subroutine of the approximation algorithm. To obtain algorithms with better performance, one could, for instance, change the exhaustive listings of the clusters by a scheme involving importance sampling over them, but the exact method to do it is currently unclear \cite{chen2021fastalgorithmslowtemperatures,blanca2023fastperfectsamplingsubgraphs}.

\begin{thm}\label{th:algorithm}
Let $\beta < \beta^*$. There exists an algorithm that outputs a function $f_\beta$ such that $\left \vert \log Z- f_\beta \right \vert \leq \epsilon$  in subexponential time,
    \begin{equation}
        \mathrm{e}^{\mathcal{O}(\log^2(N/\epsilon))}.
    \end{equation}
\end{thm} 
\begin{proof}
    All clusters of size at most $m$ can be listed in time $N^{\mathcal{O}(km)}m!$ via Lemma~\ref{lemma:alg_countingclusters}. Then for each of these clusters, we can compute the Ursell function and the polymer weights $w_\gamma$ in $\exp{\mathcal{O}(m)}$ time by Lemmas 6 and 7 of~\cite{efficient_alg_Mann_2021} respectively (see also~\cite{Haah2024}). Therefore, each term in the series of $T_m$ can be computed in $N^{\mathcal{O}(k m)}m!$ steps. By choosing $m=\mathcal{O}(\log(N/\epsilon)/\log(\beta^*/\beta))$, it follows from Theorem~\ref{Theorem:KoteckyLR} that $T_m$ approximates $\log Z_\rho$ to $\epsilon$ precision, which can be computed exactly in $\mathcal{O}\left(N^{\mathcal{O}(\log N / \epsilon)} (\log (N/\epsilon)!\right)= \mathrm{e}^{\mathcal{O}(\log^2(N/\epsilon))}$  steps.
\end{proof}
Previous work on rigorous approximation algorithms for long range models is restricted to the 1D case for $\alpha>2$ covered in~\cite{achutha2024efficientsimulation1dlongrange}. This shows a tensor network algorithm for the Gibbs state at any temperature $\beta$, also with subexponential runtime $\mathrm{e}^{\mathcal{O}(\beta \log^2(N/\epsilon))}$. Our findings are complementary, as we cover arbitrary dimensions, beyond 1D, but restricted to high temperatures.

Note that if we only require a weaker approximation error, such that we output a function such that $\vert \log Z_\rho-f_\beta \vert \leq \epsilon' N$ (as defined in \cite{Bravyi_2022}), we reduce the time complexity of our algorithm to $\exp\!\big(\mathcal{O}(\log^2(\epsilon'^{-1}))\big)$, independent of system size. Notice that it is done by simply redefining the desired error $\epsilon \rightarrow \epsilon' N$.
\section{Statistical properties}\label{sec:statistical}
Given a quantum state $\rho$ the distribution of the possible outcomes of measuring an observable $A$ is
\begin{equation}
    p_{A, \rho} (a) =  \Tr [ \Pi_a \rho],
\end{equation}
where $\Pi_a$ is the projector onto the subspace of eigenstates of eigenvalue $a$. The order-$m$ moments of the distribution are $ \langle  A^m\rangle_\rho=\Tr [  A^m \rho]$, and the moment generating function of this distribution is defined as
\begin{equation}
    M_ A \left(\tau \right)=\Tr \left[  \mathrm{e}^{\tau A} \rho\right].
\end{equation}
We can control these functions thanks to Theorem~\ref{theo:KP_GenPartFunct} if $\rho$ is a Gibbs state or a product state. This enables us to elucidate statistical properties of long-range systems.

\subsection{Concentration bounds and ensemble equivalence}

\label{sec4_concentrationbounds}

We now show the Chernoff-Hoeffding concentration bound of the tail of the distribution $p_{\beta,A} (a)$, for Gibbs states $\rho = \mathrm{e}^{-\beta H}/Z$.

\begin{coro}
Let $\alpha>D$ and $\beta  \leq\beta^*/2$. It holds that
  \begin{equation}
P_{A, \beta} (|a - \langle A \rangle_{\beta} | > \delta) \leq 2 \exp \left( - \frac{\delta^2}{4 c_\beta N} \right),
\end{equation}
where $c_\beta=2u \left(\frac{1}{\beta^*-\beta}\right)^2$. 
\end{coro}
Thus, for macroscopically large deviations \( \delta \propto N \), the probability of measuring \( A \) to be away from \( \langle A \rangle_{\beta} \) by at least \( \delta \) is exponentially small. 

The bound follows from the inequality on the moment-generating function $M_A (\tau)$
\begin{equation}
\log M_A (\tau) =\log \langle \mathrm{e}^{\tau (A - \langle A \rangle_{\beta})} \rangle_{\beta} \leq c_\beta \tau^2 N,
\label{eq::tau2}
\end{equation}
whose derivation using Theorem~\ref{theo:KP_GenPartFunct} is left to Appendix~\ref{app:CHbound}. An analogous result also holds for $\rho$ product state via Theorem~\ref{Theorem:KoteckyLR}. This result was previously reported in~\cite{Kuwahara_2020_Gaussian}, and worked for high temperatures $\beta \le (8 \mathrm{e}^3 k gu)$, including long-range models \footnote{The convergence proof of the cluster expansion it relies on has recently been found to be flawed (see~\cite{Haah2024,classsimofshort_PRXQuantum.4.020340}), although a fix has been reported~\cite{tomotaka-private}.}. 

One of the main consequences of this bound is the equivalence of ensembles between the canonical and microcanonical states. In large systems, this means that the average macroscopic properties of both the thermal (canonical) state and the microcanonical ensemble are essentially the same.  This is thus an important feature of large systems at equilibrium, which motivates the search for criteria ensuring the equivalence of the two ensembles. This was known to hold for finite-range systems~\cite{Lima1971, MullerAdlamd2015,Kuwahara_2020_ETH}. Since ensemble equivalence can be shown to hold from the Chernoff-Hoeffding bound under rather general conditions~\cite{Kuwahara_2020_Gaussian,QMBD_inTE_reviewAA}, our main result implies the following.
\begin{coro}
    For long-range interacting systems, ensemble equivalence holds when \( \alpha > D \) and \( \beta \leq \beta^*/2 \).
\end{coro}
However, ensemble inequivalence often arises in the strong long-range regime \( \alpha \leq D \)~\cite{Barr__2001,Kastner_2010,Russomanno_2021,Defenu_2024}.
The corollary then shows that, at least at high temperatures, the condition $\alpha \leq D$  is necessary for ensemble inequivalence. This also means that we do not expect our results to extend to the strong long-range regime.

\subsection{Gaussianity of the Density of States}
\label{sec2_Gaussianity}

By choosing a purely imaginary variable $\tau=i \lambda$ in the moment-generating function we recover the characteristic function
\begin{equation}
    \varphi_ A (\lambda)= \Tr [ \mathrm{e}^{i\lambda A} \rho] \,,
\end{equation}
of the distribution $p_{A, \rho} (a)$. The control we have over it due to Theorems~\ref{Theorem:KoteckyLR} and~\ref{theo:KP_GenPartFunct} allows us to prove further statements of probability theory, akin to the central limit theorem.

\begin{thm}
Consider the cumulative distribution function
\begin{equation}
    C(x)=\int_{-\infty}^x dy \, p_{A, \rho}(y), 
\end{equation}
where $\rho$ is either a product or Gibbs state at $\beta < \beta^*/2$, and the corresponding Gaussian cumulative function
      \begin{equation}
        \mathcal{G}(x)=\int_{-\infty}^x dy \frac{1}{\sqrt{2 \pi \sigma_\rho^2}} \mathrm{e}^{-\frac{\left(y-\langle A\rangle_\rho\right)^2}{2 \sigma_\rho^2}},
    \end{equation}
 with $\sigma_\rho^2= \langle  A^2\rangle_\rho - \langle  A\rangle_\rho^2 \geq \Omega(N^{1/2}) $. Their distance, defined as
 \begin{equation}
        \zeta_N = \max_x |C(x)-\mathcal{G}(x)|
    \end{equation}
is bounded by
 \begin{equation}
    \zeta_N \leq \mathcal{O}(N^{-1/2}).
    \end{equation}
    \label{th:LR_BEtheorem}
\end{thm}
\begin{figure}[t]
    \centering
    \includegraphics[width=\linewidth]{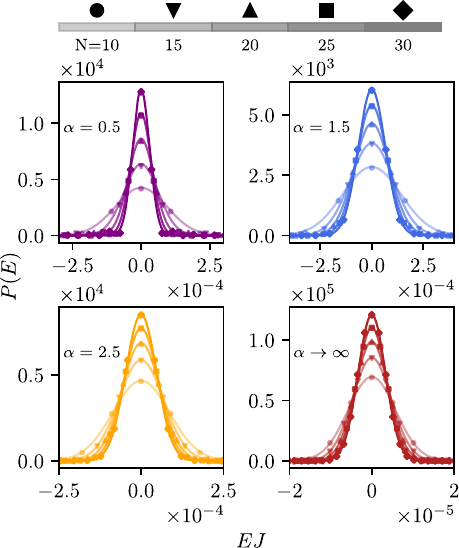}
    \caption{\justifying The density of states of the long-range transverse-field Ising model Eq.~\eqref{eq:def-TFI} varying $N$ for $\alpha=0.5, \, 1.5, \, 2.5$ and the short-ranged transverse-field Ising model ($\alpha \rightarrow \infty$). We show both the numerical estimation (markers) and the fit to a Gaussian curve (solid lines). The variance decreases with system size due to the rescaling of the Hamiltonian $H \rightarrow H/\norm{H}$.}
    \label{fig:DOSandCDOSvaryingAlpha}
\end{figure} 
This shows that the measurement statistics follow a Gaussian distribution when the system size is sufficiently large. The proof is left to Appendix~\ref{app:BE}. The assumption that the variance $\sigma_\rho $ of extensive observables growing at least $\propto N^{1/2}$ is a typical one, and rules out simple cases such as states $\rho$ that are simultaneously product and eigenstates of $A$. 

This result has further direct implications for ensemble equivalence,  thermalization and other related facts~\cite{brandao2015equivalencestatisticalmechanicalensembles,bertoni2024typicalthermalizationlowentanglementstates, Rai_2024}. 

A case of particular interest is the energy distribution $p_{E, \rho}$ for which we can simply choose $A=-H$. For $\rho = I/d$ this reduces to the familiar Density of States (DOS) 
\begin{equation}
    P(E)= \frac{1}{d} \sum_i \delta(E-E_i).
    \label{eq:dos}
\end{equation}
To better understand the Gaussianity of the DOS we performed tensor networks numerical simulations of the long-range transverse-field Ising (TFI) chain,
\begin{equation}
    H_{\text{LR, Ising}}= J \sum_{i < j}^N \frac{\sigma_i^x \sigma_{j}^x}{|i-j|^{\alpha}} + h \sum_{i=1}^N \sigma_i^z \label{eq:def-TFI} \,.
\end{equation}
This allows us to check whether the Gaussianity holds for smaller $\alpha \le D$, beyond the regime of applicability of our results.

As it can be observed in Fig.~\ref{fig:DOSandCDOSvaryingAlpha}, the numerical estimation of the DOS is apparently well captured by a Gaussian fit. 
\begin{figure}[t]
    \centering
    \includegraphics[width=\linewidth]{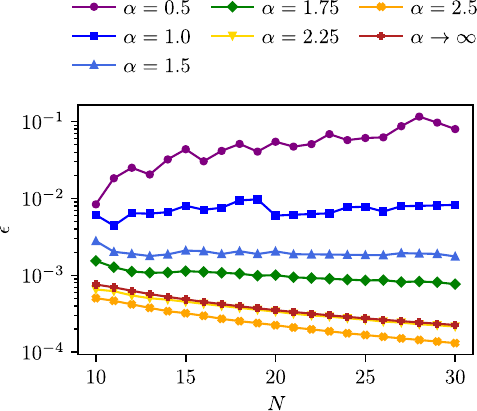}
    \caption{\label{fig:tvalue_varyingAlpha}
    \justifying
    Representation of the residuals $\epsilon = \frac{1}{R \max_i{y_i}} \sum_i^R |y_i -g(x_i)|^2$, where $g(x)$ is the Gaussian fit, as a function of system sizes and for different values of $\alpha$. For more details on the simulations see Appendix~\ref{app: sim}.}
\end{figure}
In Fig.~\ref{fig:tvalue_varyingAlpha} we plot the error of those Gaussian approximations, and find that for $\alpha>D=1$ the distributions are closer to Gaussian as $N$ increases, as expected from Theorem~\ref{th:LR_BEtheorem}. Furthermore, the error becomes smaller as the range of interactions decreases with $\alpha$. However, the error starts growing with $N$ for the cases of $\alpha \le D$, which is consistent with the expectation that the condition $\alpha>D$ is necessary for Theorems~\ref{Theorem:KoteckyLR} and~\ref{theo:KP_GenPartFunct} to hold.
\section{Absence of criticality}
We now establish how our results constrain the appearance of various forms of criticality in long-range models.  This is interesting from a fundamental viewpoint, since critical points separate regimes with different universal behavior \cite{Sachdev_2011}. The presence of criticality is also interesting for quantum simulation, since the diverging correlation lengths associated with it often pose a significant challenge for simulation methods \cite{PhysRevB.80.235127,PhysRevB.86.075117}.

\label{sec3_absenceofcriticality}
\subsection{Thermal and Dynamical Phase Transitions}
Since Theorem~\ref{Theorem:KoteckyLR} means that the log-partition function of long-range models is analytic in the thermodynamic limit for $\beta \leq \beta^*$, this rigorously implies the absence of high-temperature thermal phase transitions for long-range models, in analogy with the finite-range case. This covers a wider range of $\alpha$ (although a much narrower range in $\beta$) than the long-range extensions of the Mermin-Wagner theorem~\cite{Bruno_2001}. This is consistent with previous numerical results (e.g.~\cite{LRreview_Defenu_2023,Zhao2023}), and it also includes BKT-type transitions such as those recently observed in long-range $2D$ systems~\cite{KTPexample2_Giachetti_2021,BKT2PRB,KTPexample1_xiao2024twodimensionalxyferromagnetinduced}. For these, the critical temperature $\beta_{\text{BKT}}$ is estimated to be two orders of magnitude away from $\beta^*$~\cite{KTPexample1_xiao2024twodimensionalxyferromagnetinduced}.

This absence of thermal criticality can often be associated with other physical phenomena, such as the decay of correlations~\cite{Harrow_2020}.
For long-range systems, this decay has been proven in high temperatures in~\cite{kim2024thermalarealawlongrange}, and $1D$ at any temperatures for $\alpha>2$ in~\cite{kimura2024clusteringtheorem1dlongrange}. 

Dynamical phase transitions (DPTs) can be understood as a real-time counterpart of thermal phase transitions~\cite{DQPTs1_Halimeh_2017,DQPTs2_PhysRevB.104.115133,DQPTs3_unkovi__2018}. Since we can choose $\rho$ in Theorem~\ref{Theorem:KoteckyLR} to be any product state, our results constrain the time at which they can appear. Consider an infinite sequence of long-range Hamiltonians \( H_n \) acting on \( n \) particles, under the assumptions of the setup, and product states \( |\Psi\rangle = \bigotimes_{i=1}^{n} |\phi_i\rangle \), such that  
$
g_n(t) \equiv \log \langle \Psi | \mathrm{e}^{-\mathrm{i}tH_n} | \Psi \rangle.$
A DPT is said to occur when the following function becomes nonanalytic 
$
G(t) = \lim_{n \to \infty} \frac{g_n(t)}{n}.
$

Again as a direct consequence of Theorem~\ref{Theorem:KoteckyLR}, the function \( G(t) \) is analytic for \( t < \beta^*  \), implying that DPTs can only occur at later times.  
\subsection{Excited State Quantum Phase transitions}
A less-studied type of phase transition in quantum systems are Excited State Quantum Phase Transitions (ESQPTs). These transitions are characterized by singularities in the DOS at specific critical energies, indicating a fundamental shift in the structure of a quantum system's excited-state spectrum. Considering the definition of the DOS in Eq.~\eqref{eq:dos}, they are defined in the thermodynamic limit as follows~\cite{ESQPTs_Cejnar_2021,ESQPTs_PhysRevE.78.021106,ESQPTsPhysRevA.94.012113}.
\begin{Def}
An ESQPT occurs at energy density $e_c$ if $p(e)\equiv \lim_{N \rightarrow \infty} P(e N)$ is nonanalytic at $e=e_c$.
\end{Def}
This type of phenomenon has been found to appear in many-body spin systems with a well-defined classical limit of few degrees of freedom. These typically have permutation-symmetric Hamiltonians defined through collective spin operators, such as the LMG~\cite{Ribeiro_2007} or Dicke~\cite{Brandes_2013} models.  
In the review~\cite{ESQPTs_Cejnar_2021}, it is left open whether similar phenomena occurs in quantum systems with many degrees of freedom, such as those with a local structure, arguing that existing methods are unable to identify ESQPTs.

We can critically examine this open question by considering Theorem~\ref{th:LR_BEtheorem}, which implies that the DOS of many-body systems with long-range (but sufficiently decaying) interactions is Gaussian. For the thermodynamic limit, we can consider a renormalization Hamiltonian $H^* \equiv H/ \norm{H}$, and standard deviation $\sigma^* \equiv \sigma / \norm{H}$. As we increase $N$, the DOS becomes sharper with system size, as $\lim_{N\to \infty} \sigma^* =0$ (see Fig.~\ref{fig:DOSandCDOSvaryingAlpha} for a numerical illustration). Due to Theorem~\ref{th:LR_BEtheorem}, we can conclude that, defining $h \equiv \lim_{N\rightarrow \infty} \tr{\frac{H^*}{d^N}}$, we have
\begin{equation}
    p(e)= \delta(e-h).
\end{equation}

The only singularity is the one at $e=h$, since otherwise the DOS is trivial due to the overwhelming majority of the eigenstates concentrating at the average. Therefore, ESQPTs are seemingly trivialized by the Gaussianity of the DOS in both finite range and long range spin systems for $\alpha>D$. For $\alpha \le D$, the numerical results of Fig.~\ref{fig:DOSandCDOSvaryingAlpha} also suggest a similar picture, even if we do not expect Theorem~\ref{th:LR_BEtheorem} to hold there (see Fig.~\ref{fig:tvalue_varyingAlpha}).
\section{Conclusion} \label{sec6_conclusions}
We have shown how partition functions of long-range models at high temperatures, and generalizations thereof, behave similarly to those of finite-range models as long as $\alpha>D$. This is possible due to the convergence of the cluster expansion, whose proof is more technically involved due to the nature of the long-range interactions. We also provide numerical evidence illustrating the fact that extensions to $\alpha \le D$ are unlikely to hold.

The main quantitative difference of the present work with respect to the finite-range setting is the runtime of the approximation algorithm of Theorem~\ref{th:algorithm}. While for finite range it is clear that the algorithms are fully polynomial, in both $N$ and $\epsilon^{-1}$, here the fact that interactions are highly nonlocal means that we can only obtain superpolynomial runtimes. 

It is currently unclear whether rigorous polynomial time classical algorithms are possible, even in 1D~\cite{achutha2024efficientsimulation1dlongrange}. However, it is known that efficient quantum algorithms for the same task exist, provided that one can efficiently prepare the associated Gibbs states $\rho_\beta$ (which is not possible under general conditions). This is because if one has access to polynomially many samples of $\rho_\beta$, one can additively approximate $\log Z$ efficiently~\cite{Poulin_2009,Bravyi_2022}. The efficient preparation was shown in~\cite{rouzé2024optimalquantumalgorithmgibbs} under stronger restrictions than those of Theorem~\ref{Theorem:KoteckyLR} on the temperature and interaction ranges, although quantum efficiency is expected to hold more widely.

This situation might suggest a superpolynomial separation between classical and quantum algorithms for estimating partition functions of long-range models. However, it is possible that such separation is simply due to the lack of classical algorithms with better run-times than Theorem~\ref{th:algorithm}. One possible path toward them is to draw inspiration from classical Monte-Carlo schemes for long-range Ising models~\cite{LongRangeMonteCarlo,Flores-Sola2017,chen2021fastalgorithmslowtemperatures,blanca2023fastperfectsamplingsubgraphs}.

Our work paves the way to extending to long-range systems other important features already known for high-temperature Gibbs states of finite-range Hamiltonians. It would be particularly interesting to show a relation between the convergence of the cluster expansion and the decay of correlations~\cite{Harrow_2020} or if the Gibbs states are separable at high enough temperatures ~\cite{bakshi2025hightemperaturegibbsstatesunentangled}. Some other examples are tensor network approximations~\cite{Kliesch_2014,Molnar_2015}, the Markov structure~\cite{kuwahara2024clusteringconditionalmutualinformation,chen2025quantumgibbsstateslocally,kato2025clusteringconditionalmutualinformation,capel2024conditionalindependence1dgibbs} or efficient Hamiltonian learning~\cite{Haah2024,bakshi2023learning,chen2025learningquantumgibbsstates}.

It also has potential implications for the classical simulability results of local dynamics~\cite{classsimofshort_PRXQuantum.4.020340,mcdonough2025liebrobinson}. It would be interesting to study the relation of the convergence of $\log Z_\rho$ with the Lieb-Robinson bounds~\cite{Nachtergaele_2006,LRB_1PhysRevX.10.031010,LRB2_PhysRevLett.127.160401}.
\medskip
\begin{acknowledgments}
The authors thank Alejandro Gonz\'alez Tudela for suggesting the topic of long-range models, and Samuel Scalet and Daniel Stilck França for suggesting Lemma~\ref{lemma:alg_countingclusters}.
The work of J.S.S. was supported by the Spanish Research Council, CSIC, through the JAE-PREDOC 2024 and JAE-INTRO 2023 programs.
J.T.S acknowledges funding from the Grant TED2021-130552B-C22 funded by MCIN/AEI/10.13039/501100011033 and by the ``European Union NextGenerationEU/PRTR''.
A.M.A. acknowledges support from the Spanish Agencia Estatal de Investigacion through the grants ``IFT Centro de Excelencia Severo Ochoa CEX2020-001007-S'' and ``Ramón y Cajal RyC2021-031610-I'', financed by MCIN/AEI/10.13039/501100011033 and the European Union NextGenerationEU/PRTR. 
This project was funded within the QuantERA II Programme that has received funding from the EU’s H2020 research and innovation programme under the GANo. 101017733..
\end{acknowledgments}

\bibliographystyle{apsrev4-2}
\bibliography{references}
\clearpage
\onecolumngrid

\begin{figure*}[t]
    \centering
    \includegraphics[width=0.45\linewidth]{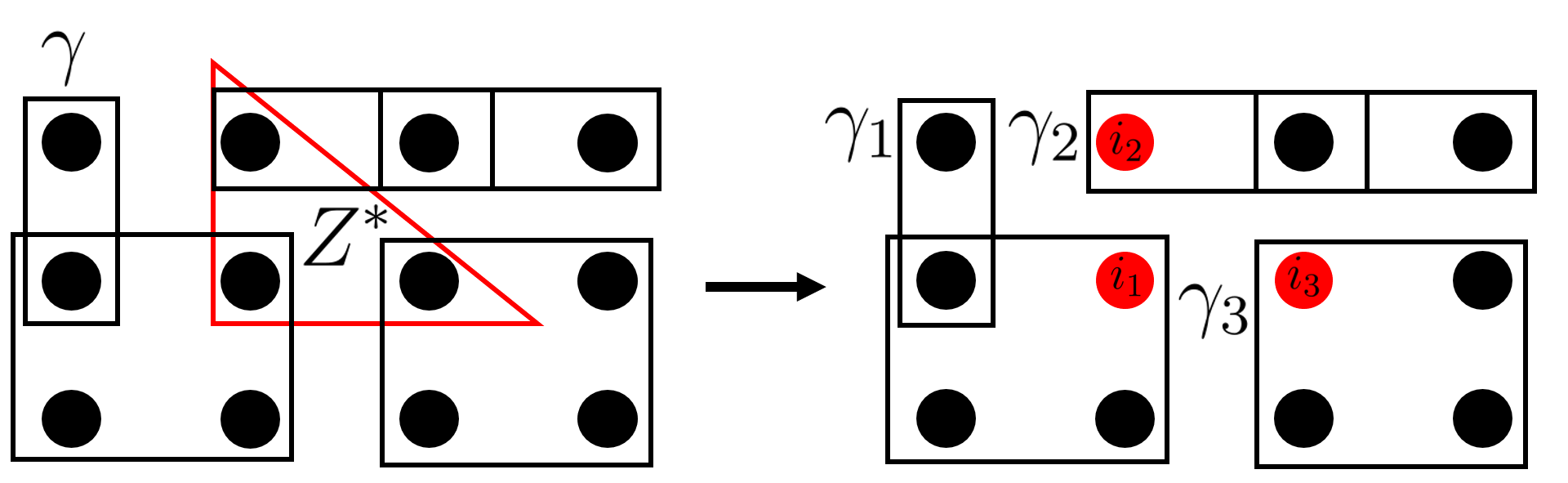}
    \caption{\justifying Scheme of a polymer containing a hyperedge $Z^*$. It has multiplicity $m=\norm{\gamma}=5$, and it is formed by 5 hyperedges. It can also be seen as a  set of three subpolymers $\gamma_1, \gamma_2, \gamma_3$ connected to sites $i_1, i_2, i_3$.}
    \label{fig:esqueme_polymer_LRproof}
\end{figure*}
\twocolumngrid

\appendix
\setcounter{footnote}{0}
\renewcommand{\thefootnote}{\arabic{footnote}}
\section{Proof of convergence criteria for long-range interacting systems}\label{sec:KoteckyProof}
\subsection{Convergence of the partition function}\label{app:KP_LR}

For a given polymer $\gamma$, we consider the function
\begin{equation}
    a \left(\gamma \right) = \vert \gamma \vert
\end{equation}
We aim to show that, for every polymer $\gamma^*$ 
\begin{equation}
    \sum_{\gamma \nsim \gamma^* } \vert w_\gamma \vert \mathrm{e}^{a\left(\gamma\right)} \leq a \left(\gamma^*\right).
\end{equation}
With this goal in mind we select a specific hyperedge $Z^*$ and, using that $\vert w_\gamma\vert \le \vert\beta|^{\norm{\gamma} } \prod_{Z \in \gamma } \norm{h_Z} $,
\begin{align}\label{eq:boundKP}
    \sum_{ \gamma \nsim Z^* } \vert w_\gamma\vert \mathrm{e}^{a\left(\gamma\right)} &= \sum_{\gamma \nsim Z^* }\vert w_\gamma\vert \mathrm{e}^{\vert\gamma\vert}  \\
    &\leq\sum_{\gamma \nsim Z^*  }\vert\beta\vert^{\norm{\gamma}} \mathrm{e}^{|\gamma|} \prod_{Z \in \gamma } \norm{h_Z} \\ 
      &\leq \sum_{\gamma \nsim Z^* } \left( \vert \beta \vert \mathrm{e}\right)^{\norm{\gamma}} \prod_{Z \in \gamma} \norm{h_Z}.
\end{align}
We thus need to control this sum over all polymers, or equivalently tuples of hyperedges that are connected to $Z^*$, which has support on sites $\{i_1, \dots, i_k\}$, where $k$ refers to the $k$-body terms of the Hamiltonian. Let us order the sum in terms of the cardinality of the polymers, so that 

\begin{align}
    &\sum_{\gamma \nsim Z^* } \left(|\beta| \mathrm{e}\right)^{\norm{\gamma}} \prod_{Z \in \gamma} \norm{h_Z} =  \sum_{m} \left(|\beta| \mathrm{e}\right)^{m} \underbrace{\sum_{\substack{ \gamma \nsim Z^* \\ \norm{\gamma}=m} }\prod_{Z \in \gamma }\norm{h_Z}}_{\mathbf{S}_{Z^*}}.
    \label{eq::expansionBeforeUpperBound}
\end{align}

To upper bound $\mathbf{S}_{Z^*}$, it is convenient to divide the product into $k$ subpolymers $\{\gamma_j\}_{j=1}^k$ with multiplicities $\{m_j\}$, each connected to vertices $\{i_j\}_{j=1}^k$ . To sum over all connected polymers of cardinality $m$ connected to $Z^*$ it suffices to count over all the sets of polymers connected to the vertices $\{i_j\}$ such that $m=\sum_j m_j$. This can always be done for all $ \gamma$, see Fig.~\ref{fig:esqueme_polymer_LRproof} for an example. This leads to a (tolerable) overcounting, as different combinations of the subpolymers $\{\gamma_j\}_{j=1}^k$ can correspond to the same polymer $\gamma$. With this in mind, we write 
\begin{align}
    \mathbf{S}_{Z^*} \leq \sum_{\substack{m_1, \dots m_j, \dots m_k \\ \sum_j m_j =m}} \prod_{j}  \sum_{ \substack{\gamma \in \Upsilon_{i_j}\\ \norm{\gamma}=m_j }}  \prod_{Z \in \gamma} \norm{h_Z}, \\
    \label{eq::S_Z*_beforebound}
\end{align}
where we now sum over all the elements of the set $\Upsilon_{i_j}$ of polymers $\gamma$ connected to site $i_j$. To continue we thus need to control the sum over all polymers, or equivalently tuples of hyperedges, of certain multiplicity $m_i$ which are connected to a given site $i$,
\begin{align}
    \mathbf{S}_i  \equiv \sum_{ \substack{\gamma \in \Upsilon_{i}\\ \norm{\gamma}=m_i }}  \prod_{Z \in \gamma} \norm{h_Z}.
    \label{eq::S_i}
\end{align}

At this point it is important to see that each polymer $\gamma$ can be associated with an auxiliary graph, labeled as $\tilde{\gamma}$, which has an edge for each hyperedge of $\gamma$, and such that it connects all the hyperedges $Z$, as illustrated in Fig.~\ref{fig::5a}. We refer to it as the \emph{connectivity polymer}. As it is always possible to make this association we can, using Eq.~\eqref{eq:HkLocal}, write

\begin{align}\label{eq:conpoly}
    \mathbf{S}_{i} \leq \sum_{\substack{\tilde{\gamma} \in \tilde{\Upsilon}_i\\ \norm{\tilde{\gamma}}=m_i}} \prod_{e \in \tilde{\gamma}} \sum_{Z \nsim e} \norm{h_Z} \leq \sum_{\substack{\tilde{\gamma} \in \tilde{\Upsilon}_i\\ \norm{\tilde{\gamma}}=m_i}} \prod_{e \in \tilde{\gamma}} J_{e}.
\end{align}
where $\tilde{\Upsilon}_i$ is analogously the set of all connectivity polymers $\tilde{\gamma}$ connected to site $i$.
The sum over all polymers connected to $\tilde{\gamma}$ is captured in $J_e$, as defined in Eq.~\eqref{eq:HkLocal}. It is important to note that, as shown in Fig.~\ref{fig::5b}, a given polymer $\gamma$ can have multiple associated connectivity polymers $\tilde{\gamma}$. Since in Eq.~\eqref{eq:conpoly} we sum over all possible connectivity polymers and, within each, over all its associated polymers, this leads to an overestimation. The benefit is that we turn the sum into one over polymers defined on graphs rather than hypergraphs, which simplifies our task considerably. For simplicity, we refer to connectivity polymers $\tilde{\gamma}$ simply as polymers from now on.
\begin{figure*}[t]
\centering
\renewcommand{\thesubfigure}{\alph{subfigure}}
    \begin{subfigure}[b]{0.45\textwidth}
    \centering
    \includegraphics[width=4.5cm]{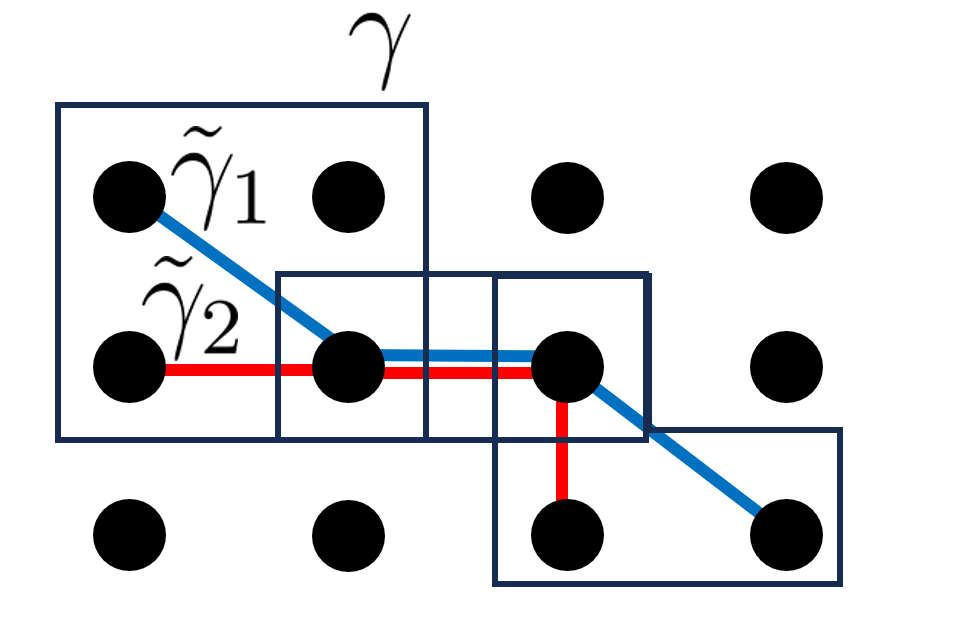}
   \caption{}
    \label{fig::5a}
    \end{subfigure}
\quad \quad
    \begin{subfigure}[b]{0.45\textwidth}
    \centering
    \includegraphics[width=4.5cm]{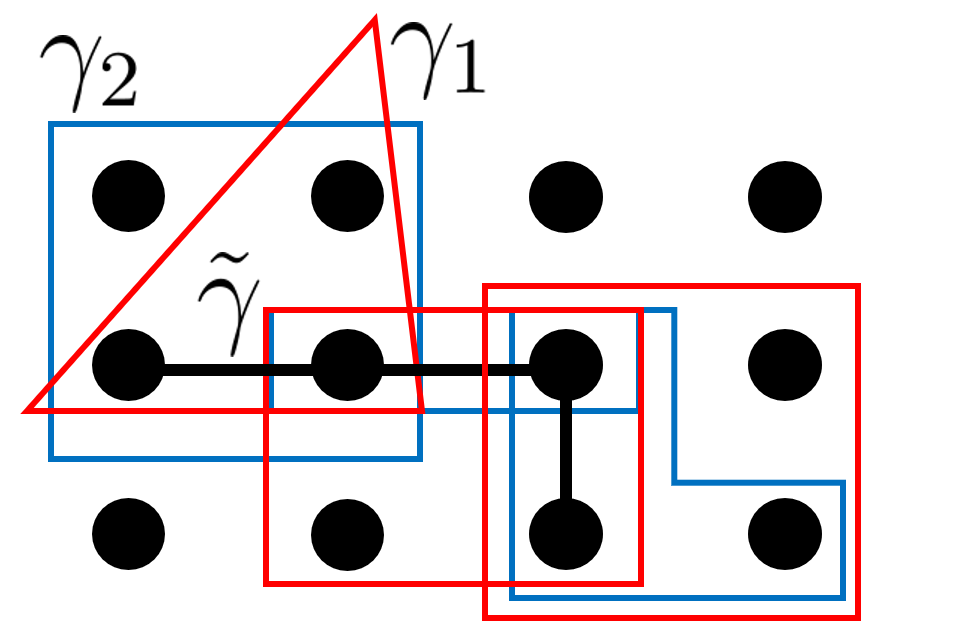}
    \caption{ }
    \label{fig::5b}
    \end{subfigure}
   \caption{\justifying \textbf{(a)} Example of a polymer $\gamma$, a tuple of hyperedges, being associated with two different connectivity polymers $\tilde{\gamma}_1$ and $\tilde{\gamma}_2$, tuples of edges. \textbf{(b)} Example of a given connectivity polymer $\tilde{\gamma}$ with different polymers $\gamma_1$ and $\gamma_2$ associated.}
\end{figure*}
At this stage it is important to remark that the notation can be expressed either in terms of edges $e$ or in terms of vertices, where $e \equiv (l,p)$ represents an edge connecting vertices $l$ and $p$. We switch between these representations as needed in subsequent steps. By writing $ \prod_{e \in \tilde{\gamma}}J_e =   \prod_{(l,p) \in \tilde{\gamma}}J_{lp}$ we have that
\begin{equation}
     \sum_{\substack{\tilde{\gamma} \in \tilde{\Upsilon}_i\\ \norm{\tilde{\gamma}}=m_i}} \prod_{e \in \tilde{\gamma}} J_{e} \equiv  \sum_{\substack{ \tilde{\gamma} \in \Upsilon_i\\\norm{\tilde{\gamma}}=m_i}} \prod_{(l,p) \in \tilde{\gamma}}J_{lp}.
    \label{eq:appA:sumUpsilon_indexes}
\end{equation}
Therefore, we are summing the product $\prod_{(l,p) \in \tilde{\gamma}}J_{lp}$ over all possible connected polymers with fixed number of edges $m_i$ of the set $\tilde{\Upsilon}_i$ of polymers connected to site $i$.
\begin{figure*}[t]
    \centering
    \renewcommand{\thesubfigure}{\alph{subfigure}} 
    \begin{subfigure}{0.45\textwidth}
        \centering
        \includegraphics[width=\textwidth]{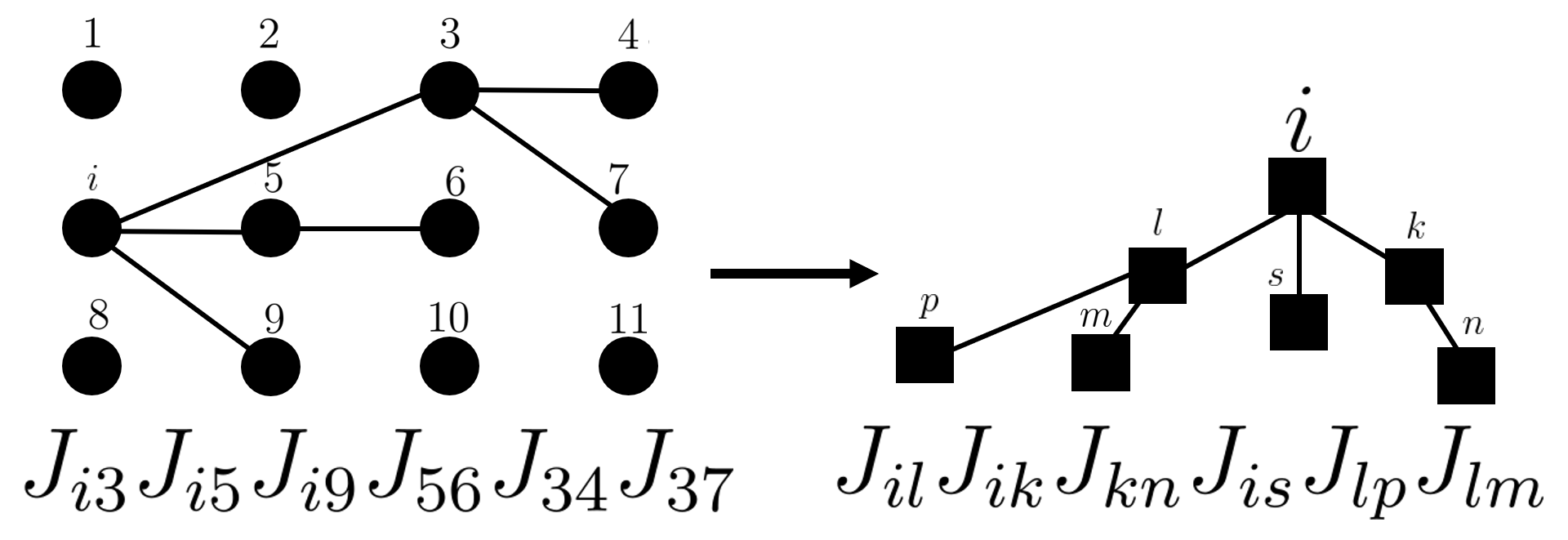}
        \caption{}
    \end{subfigure}
    \hfill
    \begin{subfigure}{0.45\textwidth}
        \centering
        \includegraphics[width=\textwidth]{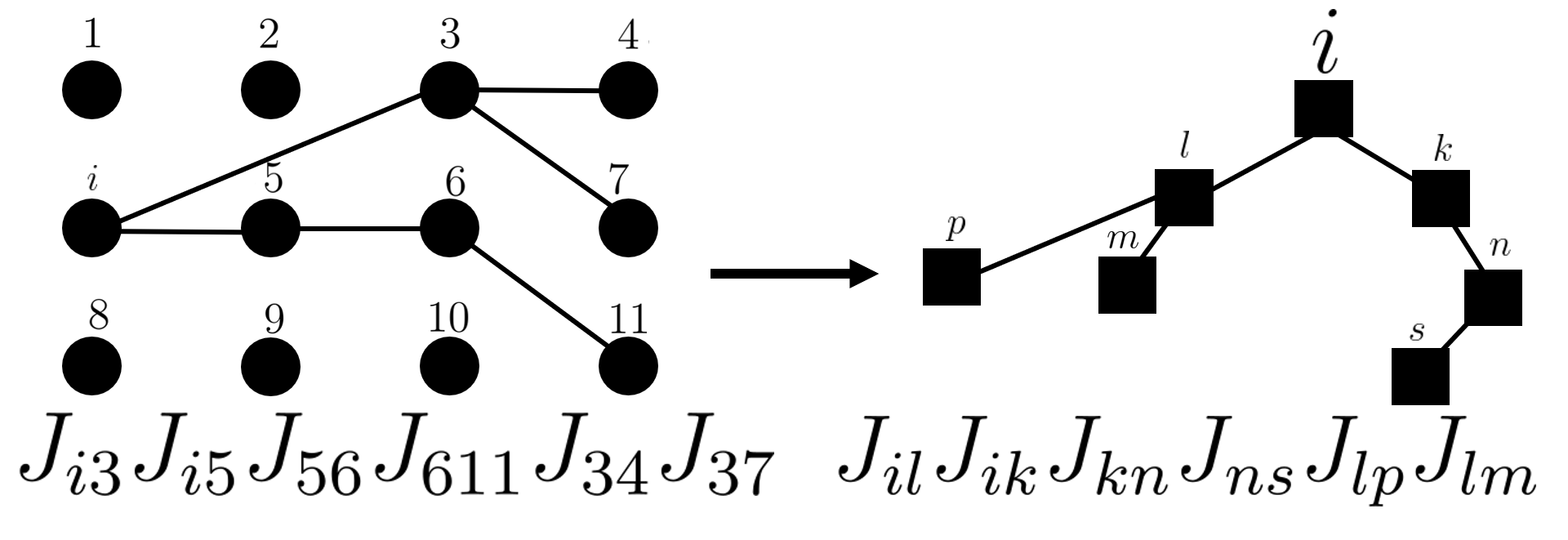}
        \caption{}
    \end{subfigure}
    
    \vspace{2pt}
    
    \begin{subfigure}{0.45\textwidth}
        \centering
        \includegraphics[width=\textwidth]{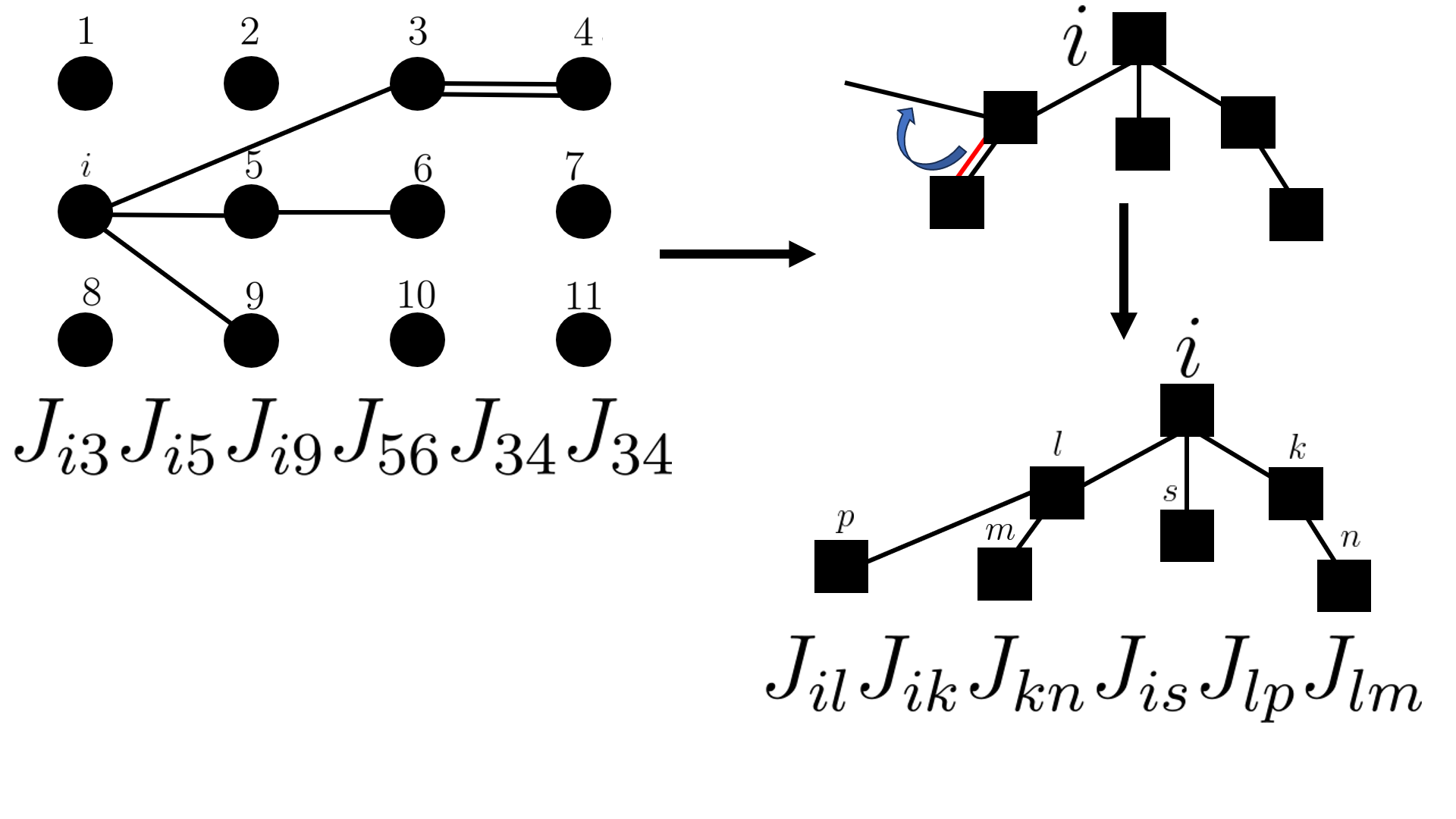}
        \caption{}
    \end{subfigure}
    \hfill
    \begin{subfigure}{0.45\textwidth}
        \centering
        \includegraphics[width=\textwidth]{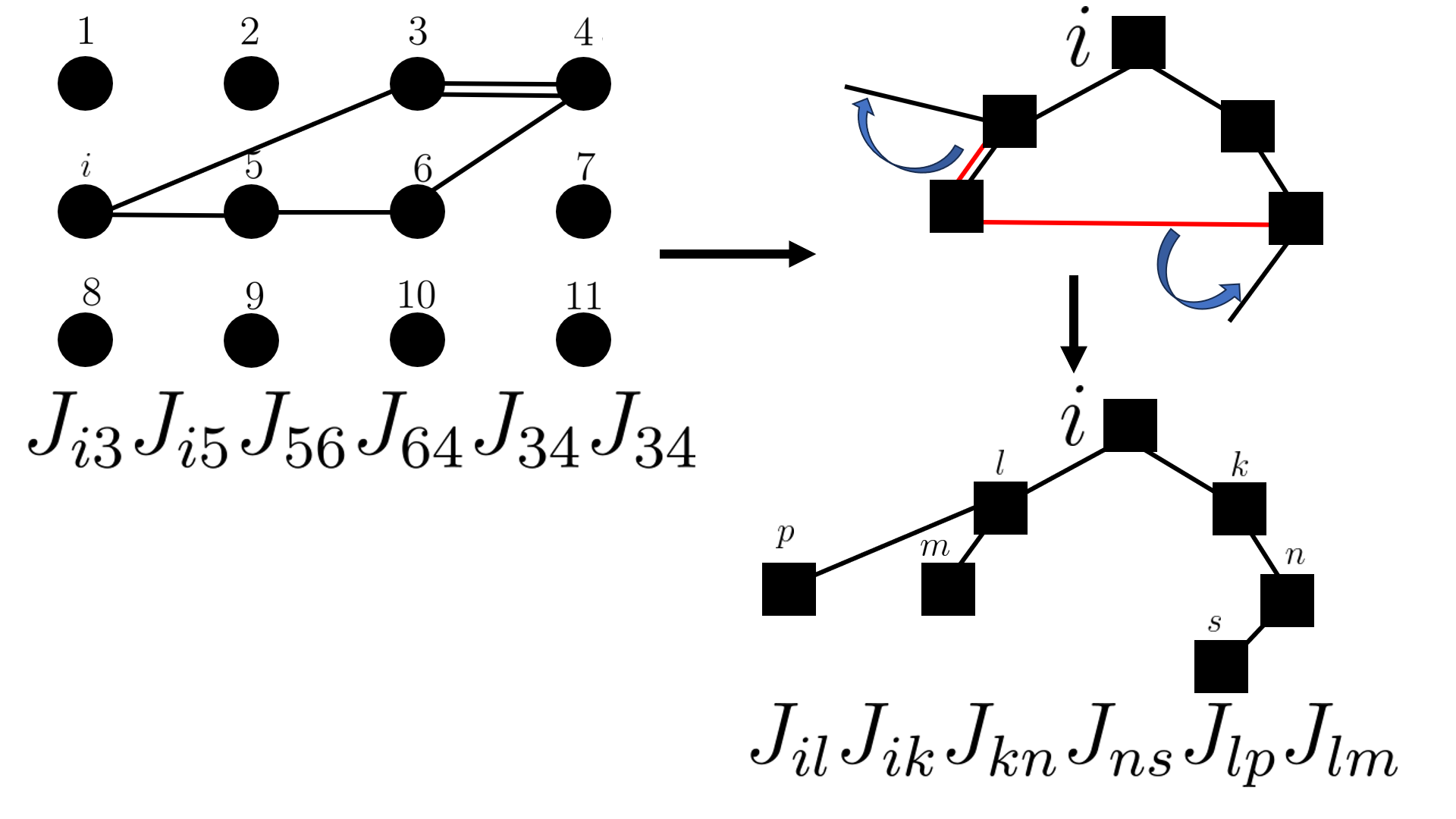}
        \caption{}
    \end{subfigure}
    \caption{\justifying \label{fig:LRlimp_Map_Lattice_to_tree}
    Illustration of the map from lattice polymer to rooted tree in indexes with some examples of $m=6$. In \textbf{(a)} and \textbf{(b)} the map is trivial, as the lattice polymer is a rooted tree itself. In cases with loops like \textbf{(c)}, the mapping can be done by adding a new auxiliary vertex for the edge not to be repeated. For \textbf{(d)}, with two loops, we add auxiliary vertices as many times as needed.
    }
\end{figure*}

\begin{figure*}[t]
\centering
\renewcommand{\thesubfigure}{\alph{subfigure}}
    \begin{subfigure}[b]{0.47\textwidth}
    \centering
    \includegraphics[width=\linewidth]{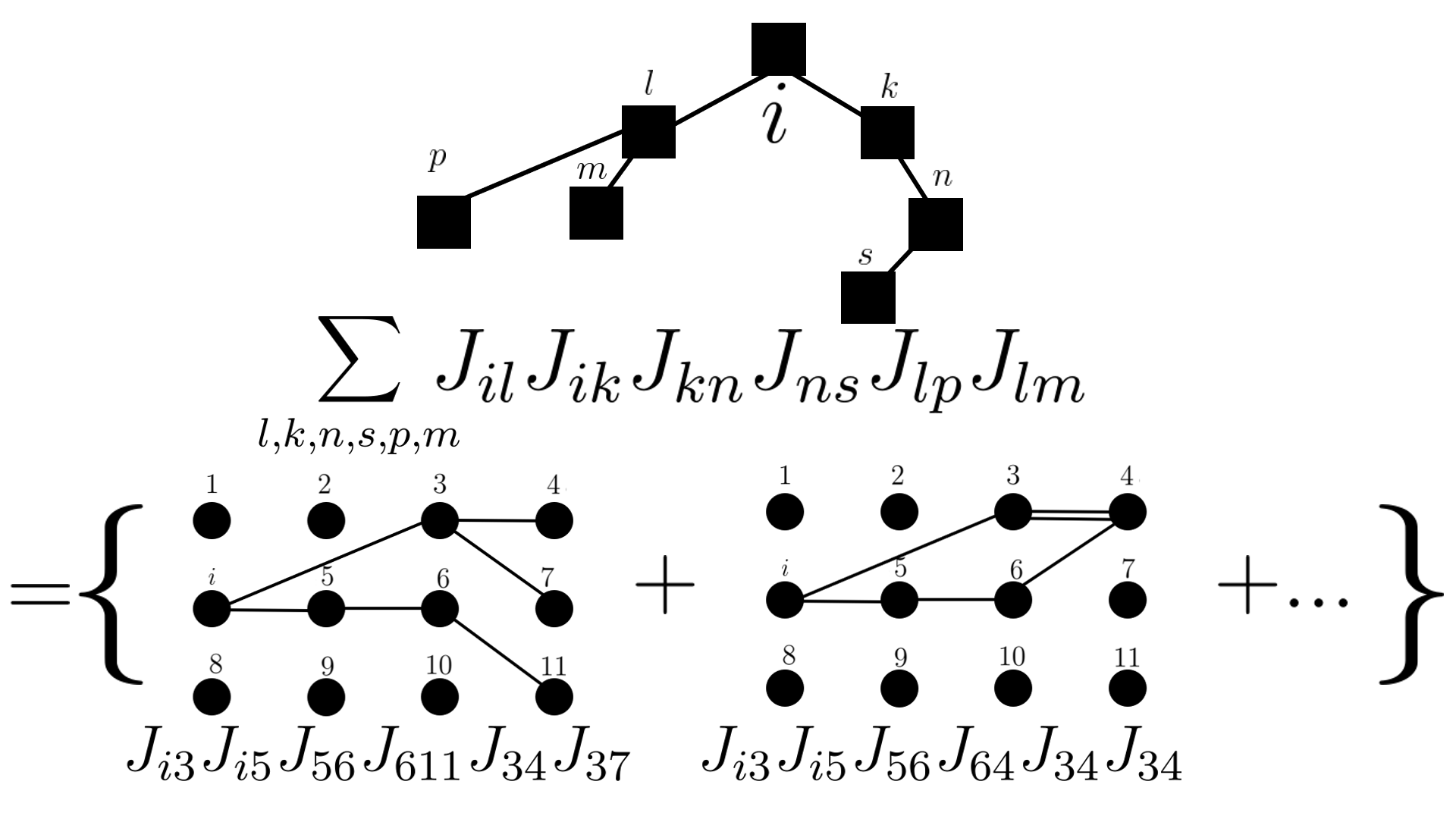}
   \caption{\label{fig::7a}}
    \end{subfigure}
\quad \quad
    \begin{subfigure}[b]{0.47\textwidth}
    \centering
    \includegraphics[width=\linewidth]{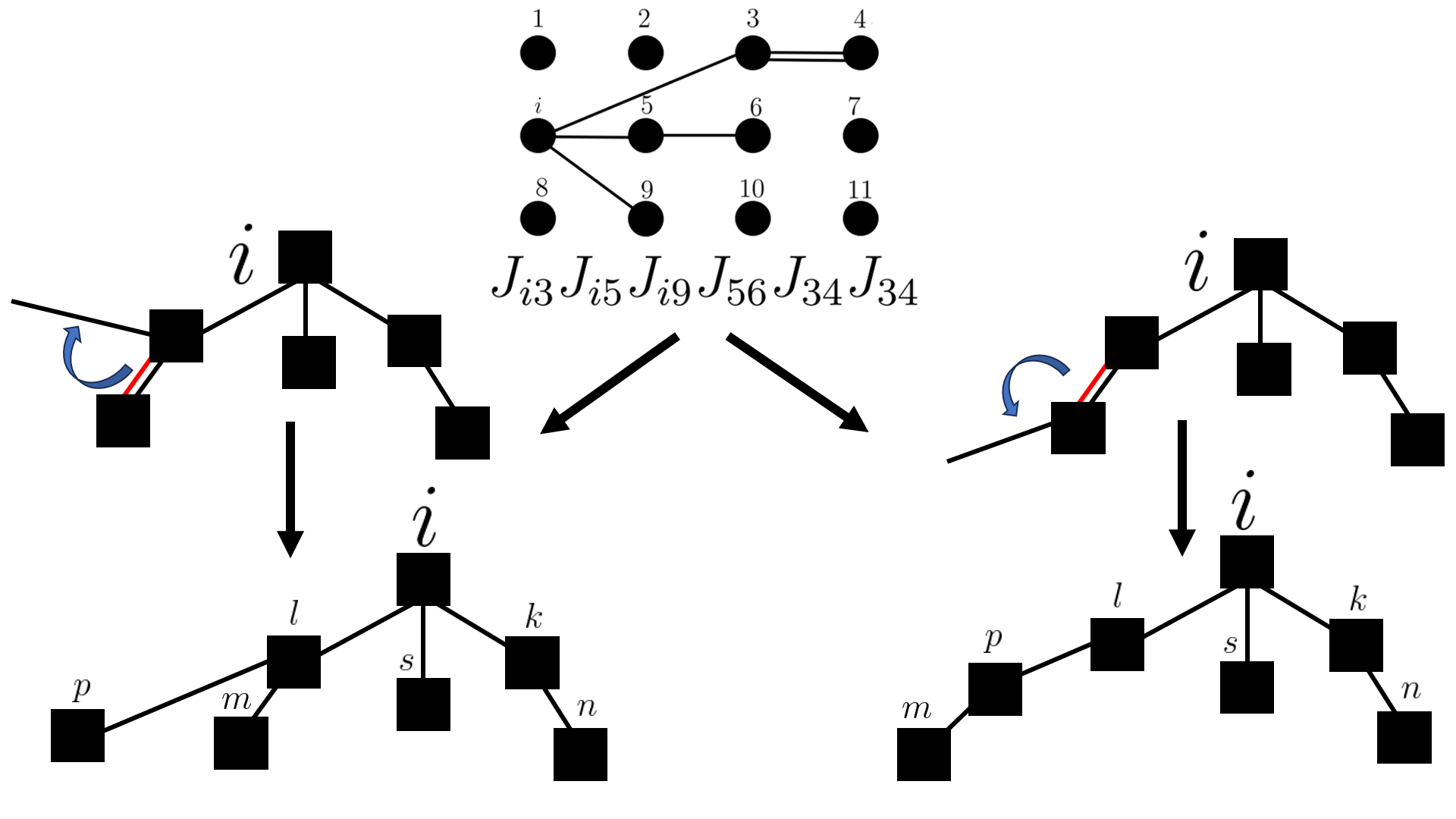}
    \caption{\label{fig::7b}}
    \end{subfigure}

   \caption{\justifying \textbf{(a)} Example of a single rooted tree being associated to different polymers. The sum of the indexes over the lattice of a certain rooted tree structure results in the sum of a set of polymers with a fixed cardinality $m$ connected to site $i$. \textbf{(b)} A lattice polymer which can be mapped to two different rooted trees. In this example, the polymer can be mapped to a different rooted tree by adding the auxiliary index in a new layer instead of the previous one. This incurs in a tolerable overcounting of some polymers.}

\end{figure*}
The key now is to realize that each term $\prod_{(l,p) \in \tilde{\gamma}}J_{l,p}$ has a structure which can be made to correspond to a rooted tree $T=\{V_T,e\}$, with vertex $i$ being the root and the arbitrary indexes $l,p,k...$ as the vertices. Any connected polymer can be assigned (in a nonunique way) to a particular rooted tree by substituting the indexes $V_T$ by vertices in the appropriate rooted tree. Illustrations of that nonunique mapping are shown in Fig.~\ref{fig:LRlimp_Map_Lattice_to_tree}.

The existence of this map implies that if we sum over all rooted trees, and each is itself summed over the sites of the entire lattice, the resulting sum is an upper bound to that over all polymers in $\tilde{\Upsilon}_i$. As illustrated in Fig.~\ref{fig::7a}, the sum takes the form
\begin{align}
      \mathbf{S_{i}} \leq  \sum_{\substack{T \\ \vert V_T \vert=m_i+1}} \, \sum_{\{V_T\}\equiv l,p,k \dots} \, \prod_{ e \in T} J_{e}.
     \label{eq:antesLemmaCommonUpperBound}
\end{align}

Notice that here we incur in a possibly important overcounting, since the same lattice polymer can be mapped to different rooted tree structures, as exemplified in Fig.~\ref{fig::7b}. The next step is to obtain a common upper bound for the sum $\sum_{\{V_T\}\equiv l,p,k...} \, \prod_{\forall e \in T} J_{e}$ independently of the tree $T$.

\begin{lemma}\label{lemma:common_upperbound}
   Given a rooted tree $T$ such that $\vert V_T \vert=n+1$, the sum over all associated polymers is such that 
    \begin{equation}
        \sum_{\{V_T\}\equiv l,p,k...} \prod_{ e \in T} J_{e} \leq (gu)^{n}.
    \end{equation}

\begin{proof}
    We first show the simple case of a star-like tree with $m_i=3$ in Fig.~\ref{fig::8a},
    \begin{align}\label{eq:stargraph}
     \sum_{\substack{T \text{ star-like}\\ m_i=3}} \prod_{e \in T}J_e &\equiv \sum_{p, l, j } J_{ij}J_{il} J_{ip} \\  &\leq\sum_{l j } J_{ij} J_{il} \sum_{p} J_{ip}\\ &\leq gu\sum_{l j } J_{ij}J_{il} \leq  \left(gu\right)^3.
     \end{align}
It is easy to see that a $n$-leaved star tree has a corresponding bound, leading to 
     \begin{equation}\label{eq:stargraph_N}
     \sum_{\substack{T \text{ star-like}\\ m_i=n}} \prod_{e \in T}J_e \leq  \left(gu\right)^n.
     \end{equation}
The complementary case is a linear graph, as the one shown in Fig.~\ref{fig::8b}
\begin{equation}\label{eq:linegraph}
     \sum_{\substack{T \text{ linear}\\ m_i=3}} \prod_{e \in T}J_e \leq \sum_{\substack{p \\ l,k}} J_{ik}J_{kl} J_{lp} = \sum_{p} \left[\textbf{J}^3\right]_{ip},
\end{equation}
where $\left[\textbf{J}^n\right]_{ij}$ corresponds to the linear graph of $m_i$ edges which starts at $i$ and ends at $l$. As $(1+d_{ij})^{-\alpha}$ satisfy the convolution condition, for any linear graph we have~\cite{Nachtergaele_2006,Nachtergaele_2019,kim2024thermalarealawlongrange} 
\begin{equation}
    \sum_{\substack{T \text{ linear}\\ m_i=n}} \prod_{e \in T}J_e = \sum_j \left[\textbf{J}^n\right]_{ij} \leq \left(gu\right)^n.
\end{equation}
From these examples we can generalize to larger tree structures by an inductive argument.
First, for all the leaves of the rooted tree that belong to a line subgraph of size $s \ge 2$, their contribution can be bounded away by $(gu)^s$ as in Eq.~\eqref{eq:linegraph} until we obtain a sum that can be mapped to a tree in which all the leaves belong to a star subgraph. Then, for each of the star subgraphs with $s$ leaves, we upper bound their contribution by $(gu)^s$ as in Eq.~\eqref{eq:stargraph_N}. We can do this until the resulting tree has no more edges, and we are only left with the root, as shown in Fig.~\ref{fig::8c}. Since there are $n$ edges in total, the sum is upper bounded by $(gu)^n$.
\end{proof}
\end{lemma}
\begin{figure*}[t]
    \begin{minipage}{0.4\textwidth}
        \centering
        \subfloat[][\justifying Upper bound of the summation over star-like trees of $m_i=3$ connected to i.]{
            \includegraphics[width=0.9\linewidth]{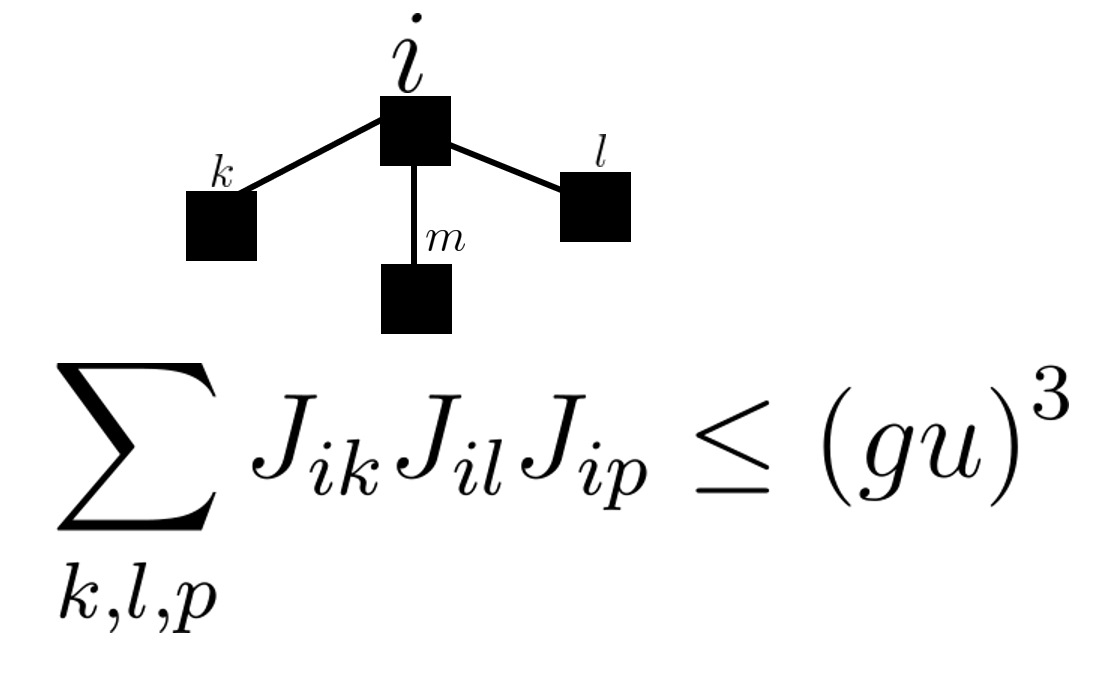}
            \label{fig::8a}
        }\\
        \subfloat[][\justifying Upper bound of the summation over linear trees of $m_i=3$ connected to i.]{
            \includegraphics[width=0.9\linewidth]{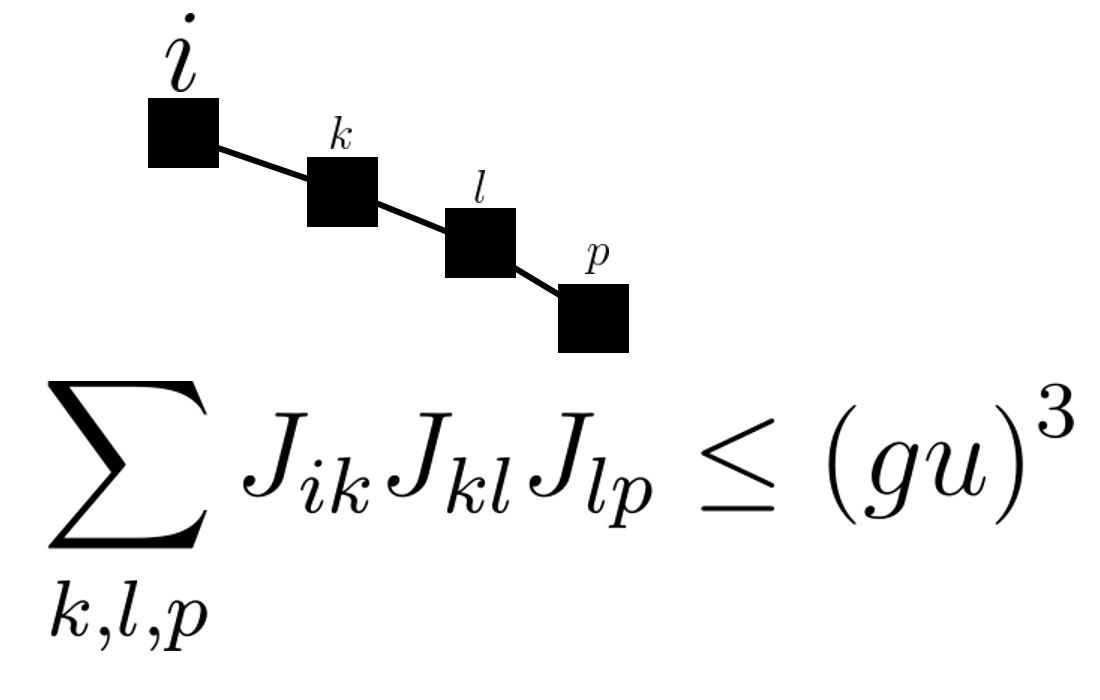}
            \label{fig::8b}
        }
    \end{minipage}
    \begin{minipage}{0.55\linewidth} 
        \centering
        \subfloat[][\justifying Upper bound of the summation over all trees with a certain structure of $m_i=6$ connected to i. We first upper bound the sequential contributions, reducing the size of the tree, and finally remove the star-like ones, until we have summed over all indexes.]{
            \includegraphics[width=1\linewidth]{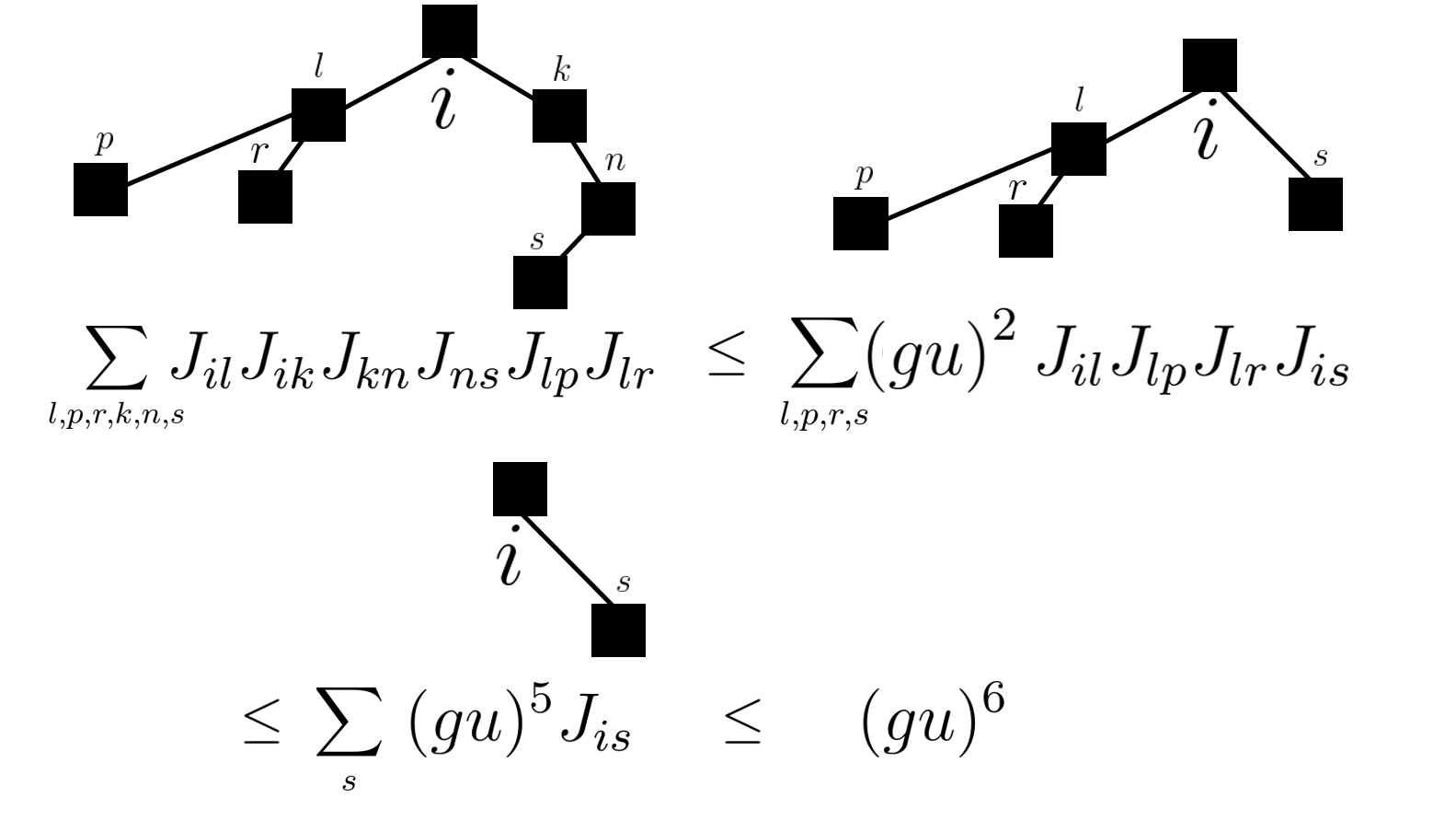}
            \label{fig::8c}
        }
    \end{minipage}
    \caption{\justifying We represent the process of upper-bounding the contribution corresponding to a rooted tree in the proof of Lemma~\ref{lemma:common_upperbound}. The upper bound depends only on the cardinality $m_i$, and not on the specific tree.}
    \label{fig:composite}
\end{figure*}
We conclude from Lemma~\ref{lemma:common_upperbound} that the contribution of each rooted tree in the sum of Eq.~\eqref{eq:antesLemmaCommonUpperBound} has an upper bound independent of the tree, so that we can write
\begin{align}
     \mathbf{S}_i \leq    \sum_{\substack{T  \\|V|=m+1}} \left(g u\right)^{m}= \left(g u\right)^{m} \sum_{\substack{T  \\|V|=m+1}} 1.
     \label{eq:despuesLemmaCommonUpperBound}
\end{align}

We have thus reduced our bound to the counting of rooted trees with a certain amount of edges.
If a tree has $m_i$ edges, it consequently has $m_i+1$ vertices. Standard combinatorial arguments show that the number of unlabeled rooted trees with $n$ vertices is lower than $4^{n-1}$~\cite{TheNumberOfTrees}, from which it follows that
\begin{equation}
    \mathbf{S}_i \leq \left( 4 g u\right)^{m_i}.
\end{equation}
Note from~\cite{TheNumberOfTrees} that there are tighter asymptotic estimates of the number of unlabeled rooted trees, so we expect that a better estimate is possible. This would result in a larger radius of convergence of the series. 

At this point, we can incorporate the bound into Eq.~\eqref{eq::S_Z*_beforebound} yielding
\begin{align}
    \mathbf{S}_{Z^*} &\leq \sum_{\substack{m_1, \dots m_j, \dots m_k \\ \sum_j m_j =m}} \prod_{j} \left(4gu \right)^{m_j} \\&= \left(4gu\right)^m \sum_{\substack{m_1, \dots m_j, \dots m_k \\ \sum_j m_j =m}} 1 = \left(4guk\right)^m \,.
\end{align}
Finally, plugging this bound in Eq.~\eqref{eq:boundKP} leads to 
\begin{align}
    &\sum_{\gamma \nsim Z^*} \vert w_\gamma\vert \mathrm{e}^{a\left(\gamma\right)} \leq  \sum_{m=1} \left(4\vert\beta\vert \mathrm{e} gu k\right)^{m}. \label{eq::expansionBeforeUpperBound_pluggingresult}
\end{align}
Then, for $\vert \beta \vert < \beta' =\left(4\mathrm{e}guk \right)^{-1}$,
\begin{align}
    \sum_{\gamma \nsim Z^*} \vert w_\gamma\vert \mathrm{e}^{a\left(\gamma\right)} \leq \left(\frac{1}{1-\frac{\vert \beta \vert}{\beta '}}-1\right),
    \label{eq::expansionBeforeUpperBound_pluggingresult_2}
\end{align}
consequently for $ \vert \beta \vert \leq  \beta'/2$,
\begin{align}
    \sum_{\gamma \nsim Z^*} \vert w_\gamma\vert \mathrm{e}^{a\left(\gamma\right)} \leq 1.
    \label{eq::ExpressionJustBeforeCriteriaSatisfied}
\end{align}
We now pick an arbitrary polymer $\gamma^*$ and sum this last expression over all its hyperedges
\begin{align}
    \sum_{\gamma \nsim \gamma^*} \vert w_\gamma\vert \mathrm{e}^{a\left(\gamma\right)} &\leq \sum_{Z^* \in \gamma^*}\sum_{\gamma \nsim Z^*} \vert w_\gamma\vert \mathrm{e}^{a\left(\gamma\right)} \\  &\leq \sum_{Z^* \in \gamma^*} 1 = \vert \gamma \vert = a \left(\gamma^*\right).
\end{align}
This shows that the conditions of Lemma~\ref{le::Kotecky-Preiss} are satisfied. If we select a site $i$ and a hyperedge $Z_i$ connected to it and apply Eq.~\eqref{eq::KoteckyPreiss2} with $\gamma^*=Z_i $,
\begin{equation}
    \sum_{\substack{\Gamma \in \mathcal{G}_C ,\, \Gamma \nsim Z_i}} \vert \varphi(\Gamma) \prod_{\gamma \in \Gamma}  w_{\gamma}\vert \leq 1.
\end{equation}

Finally, summing over all sites, we obtain a sum over all the connected clusters of the lattice
\begin{align}
     \sum_{\Gamma \in \mathcal{G}_C } \vert\varphi \left(\Gamma\right) \prod_{\gamma \in \Gamma} w_{\gamma}\vert &\leq \sum_{i}\sum_{\substack{\Gamma \in \mathcal{G}_C\\ \Gamma \nsim i}} \vert \varphi(\Gamma) \prod_{\gamma \in \Gamma}  w_{\gamma}\vert \\
     &\leq \sum_{i}1 \leq N.
\end{align}
This upper bound is valid for any $\beta \leq \beta'/2$. We are now in a position to prove the convergence of the series
\begin{widetext}
\begin{align}
    \big\vert \log{Z_\rho} - T_m \big\vert 
    % &= \bigg\vert \sum_{\substack{ \Gamma \in \mathcal{G}_C \\ \norm{\Gamma} \geq m}} \varphi(\Gamma) \prod_{\gamma_i \in \Gamma} w_{\gamma_i} \bigg\vert \\
    &\leq \sum_{\substack{ \Gamma \in \mathcal{G}_C \\ \norm{\Gamma} \geq m}} \bigg\vert \varphi(\Gamma) \prod_{\gamma_i \in \Gamma} w_{\gamma_i} \bigg\vert \\
    \label{eq::BetaDependentUpperBound_inicio}
    % &= \sum_{\substack{ \Gamma \in \mathcal{G}_C \\ \norm{\Gamma} \geq m}} \bigg\vert \varphi(\Gamma) \prod_{\gamma_i \in \Gamma} 
    % \frac{(-\beta)^{\norm{\gamma_i}}}{\norm{\gamma_i}! \prod_{Z \in \gamma_i} m_{\gamma_i}(Z)!} 
    % \operatorname{Tr} \left[ \sum_{\sigma \in S_{\norm{\gamma_i}}} \prod_{j=1}^{\norm{\gamma_i}} h(\gamma_{i, \sigma(j)}) \rho \right] 
    % \bigg\vert \\
    % &= \sum_{\substack{ \Gamma \in \mathcal{G}_C \\ \norm{\Gamma} \geq m}} \vert\beta\vert^{\norm{\Gamma}} 
    % \bigg\vert \varphi(\Gamma) \prod_{\gamma_i \in \Gamma} 
    % \frac{1}{\norm{\gamma_i}! \prod_{Z \in \gamma_i} m_{\gamma_i}(Z)!} 
    % \operatorname{Tr} \left[ \sum_{\sigma \in S_{\norm{\gamma_i}}} \prod_{j=1}^{\norm{\gamma_i}} h(\gamma_{i, \sigma(j)}) \rho \right] 
    % \bigg\vert \\
    &= \vert\beta\vert^m \sum_{\substack{ \Gamma \in \mathcal{G}_C \\ \norm{\Gamma} \geq m}} \vert\beta\vert^{\norm{\Gamma} - m} 
    \bigg\vert \varphi(\Gamma) \prod_{\gamma_i \in \Gamma} 
    \frac{1}{\norm{\gamma_i}! \prod_{Z \in \gamma_i} m_{\gamma_i}(Z)!} 
    \operatorname{Tr} \left[ \sum_{\sigma \in S_{\norm{\gamma_i}}} \prod_{j=1}^{\norm{\gamma_i}} h(\gamma_{i, \sigma(j)}) \rho \right] 
    \bigg\vert \\
    &\leq \bigg\vert \frac{2\beta}{\beta'} \bigg\vert^m \sum_{\substack{ \Gamma \in \mathcal{G}_C \\ \norm{\Gamma} \geq m }} 
    \bigg\vert \varphi(\Gamma) \prod_{\gamma_i \in \Gamma} w_{\gamma_i} \left[\beta=\beta'/2\right]  \bigg\vert \leq \bigg\vert \frac{2\beta}{\beta'} \bigg\vert^m N.
    \label{eq::BetaDependentUpperBound_final}
\end{align}
\end{widetext}
By defining $\beta^*=\beta'/2$, we conclude that for $\beta < \beta^*$
\begin{equation}
\big    \vert\log Z_\rho -T_m \big \vert \leq N \left\vert \frac{\beta}{\beta^*}\right\vert^m 
\end{equation}
and consequently $\log Z_\rho $ is analytic and the cluster expansion converges.
\clearpage
\subsection{Koteck\'y-Preiss convergence criterion for the generalized partition function of long-range lattice model}\label{app:KP_GenPartFunct}

Here we focus on a slightly different object,
\begin{equation}
    \log Z_{\rho, Gen} =  \log \text{Tr} \left[\prod_{l=1}^K \mathrm{e}^{\lambda_l H_l} \rho \right],
    \label{eq:appendix_logZGen}
\end{equation}
where $\{ H_l \}$ satisfies our initial assumptions, so that we include both local and long-range operators. Now there are multiple exponentials of operators inside the trace, which will give different kinds of hyperedges, up to $K$, each labeled with $l$. Additionally there is an order in which the terms can appear in the polymer expansion, which we identify by assigning a different "color" to each kind of hyperedge.

We first need an expression for the weight $w_\gamma^t$ of each polymer. We follow a calculation similar to Appendix A of~\cite{efficient_alg_Mann_2021}, but now with an extra degree of freedom due to the color.
\begin{widetext}
\begin{lemma}
    The polymer weights $w_\gamma^t$ of the expansion of $ \log \Tr \left[\prod_{l=1}^K \mathrm{e}^{\lambda_l H_l} \rho \right] $ with $\rho$ a product state are given by
\begin{equation}
    w_{\gamma}^t= \left[\prod_l \frac{\lambda_l^{t_l^i}}{t_l ! \prod_{Z^l \in \gamma}m(Z^l)!} \right]\Tr \left[ \prod_l \left( \sum_{\sigma \in S_{t_l}}\prod_{p =1}^{t_l} h(\gamma_{\sigma (p)})\right) \rho\right],
\end{equation}
    where $t_l$ is the number of hyperedges of certain color $l$ of a polymer $\gamma$, and $\sigma$ is a permutation of hyperedges of the same color.
\end{lemma}

\begin{proof}
We start by writing the inner sum as a product over disjoint objects, which we identify with the polymers $\gamma$

\begin{align}
      Z_{\rho, Gen} &= \text{Tr} \left[\prod_l \mathrm{e}^{\lambda_l H_l} \rho\right] = \sum_n \sum_{ \substack{n_A,\dots, n_l,\dots ,n_K\\\sum_l n_l=n}} \text{Tr} \left[\prod_l \frac{\lambda_l^{n_l}}{n_l !} \left(\sum_{Z\in Z_l(G)} h_Z\right)^{n_l}\rho\right],
\end{align}
where $Z_l(G)$ is the set of hyperedges of a certain color $l$ in the lattice. Let $ S = (Z_i)_{i=1}^{n} $ be an ordered sequence of hyperedges. Within that sequence there are connected subsequences of hyperedges. We call those subsequences of connected hyperedges sequential polymers $\gamma_s$. In those sequential polymers, repeated hyperedges are considered as different, due to the ordering of the sequence. Note that these sequential polymers, as the sequence $S$, preserve an ordering of color. Let $\Gamma_S$ be the set of sequential polymers in $S$, then
\begin{align}
      Z_{\rho, Gen} &=  \sum_n \sum_{ \substack{n_A,\dots, n_l,\dots ,n_K\\\sum_l n_l=n}}  \left[\prod_l \frac{\lambda_l^{n_l}}{n_l !}  \sum_{Z_1, \dots, Z_{n_l}\in   Z_l(G)} \right] \text{Tr} \left[ \prod_l \prod_{i}^{n_l} h_{Z_i}\rho\right] \\
      &=  \sum_n \sum_{ \substack{n_A,\dots, n_l,\dots ,n_K\\\sum_l n_l=n}}  \left[\prod_l \frac{\lambda_l^{n_l}}{n_l !} \right]\sum_{S : \vert S \vert=n} \prod_{\gamma \in \Gamma_S}\text{Tr} \left[ \prod_{Z \in \gamma} h_Z\rho\right]. 
\end{align}

Define  $ \Gamma_G := \bigcup_S \Gamma_S $ as the set of all sequential polymers in $G$, and let $ \mathcal{G}_G $ denote the collection of all admissible (i.e. mutually disconnected) sets of sequential polymers in $G$. We want to transform the sum of sequences into a sum of polymers. 

Let us introduce another variable in the sum, $k$, representing the number of polymers in a given sequence.  Notice that once we fix the number of hyperedges of each color $n_l$ for a polymer $\gamma_i$, $t_l^i$, there are $\prod_l \binom{n_l}{t_l^1 \dots t_l^k}$ sequences $S$ that give rise to a certain admissible set of colored polymers. Hence,
\begin{align}
      Z_{\rho, Gen} &=  \sum_n \sum_{ \substack{n_A,\dots, n_l,\dots ,n_K\\\sum_l n_l=n}}  \left[\prod_l \frac{\lambda_l^{n_l}}{n_l !} \right]  \sum_{k=0}^n \frac{1}{k!} \sum_{\substack{t_1^1, \dots t_1^i, \dots t_1^k\\...\\t_l^1, \dots t_l^i, \dots t_l^k\\...\\t_K^1, \dots t_K^i, \dots t_K^k}}\prod_l \binom{n_l}{t_l^1 \dots t_l^k}  \sum_{\substack{\gamma_1, \dots, \gamma_k \in \Gamma_G \\ t_l^i \text{ fixed } \forall \gamma_i\\ \text{admissible}}} \prod_i^k\text{Tr} \left[ \prod_{Z \in \gamma_i} h_Z\rho\right] \\
      &= \sum_n \sum_{ \substack{n_A,\dots, n_l,\dots ,n_K\\\sum_l n_l=n}}   \sum_{k=0}^n \frac{1}{k!}   \sum_{\substack{\gamma_1, \dots, \gamma_k \in \Gamma_G \\ \text{admissible}}} \prod_i^k  \left[\prod_l \frac{\lambda_l^{t_l^i}}{t_l^i !} \right]\text{Tr} \left[ \prod_{Z \in \gamma_i} h_Z \rho\right].
\end{align}

In the last step, we have merged the sum over $\{ t_l^i \}$ with the sum over the admissible polymers. By interchanging the summation over $k$ and $n$,

\begin{align}
    Z_{\rho, Gen}&=\sum_{k=0}^\infty \frac{1}{k!} \sum_{n=k}^\infty \sum_{ \substack{n_A,\dots, n_l,\dots ,n_K\\\sum_l n_l=n}}       \sum_{\substack{\gamma_1, \dots, \gamma_k \in \Gamma_G \\ \text{admissible}}} \prod_i^k  \left[\prod_l \frac{\lambda_l^{t_l^i}}{t_l^i !} \right]\text{Tr} \left[ \prod_{Z \in \gamma_i} h_Z \rho\right]=\sum_{\Gamma \in \mathcal{G}}  \prod_{\gamma \in \Gamma}  \left[\prod_l \frac{\lambda_l^{t_l^i}}{t_l^i !} \right]\text{Tr} \left[ \prod_{Z \in \gamma_i} h_Z \rho\right].
\end{align}
Finally, by transforming the sum over admissible sets of sequential polymers into a sum over admissible sets of polymers (where now repeated hyperedges are taken into account) and summing the weights of their permutations, we obtain
\begin{align}
    Z_{\rho, Gen}&=\sum_{\Gamma \in \mathcal{G}}  \prod_{\gamma \in \Gamma}  \left[\prod_l \frac{\lambda_l^{t_l^i}}{t_l^i ! \prod_{Z^l \in \gamma}m(Z^l)!} \right]\text{Tr} \left[ \prod_l \left( \sum_{\sigma \in S_{t_l^i}}\prod_{p =1}^{t_p^l} h_{\gamma_{\sigma (p)}}\right) \rho \right],
\end{align}
where we identify $h_{\gamma_{\sigma (p)}}=h_Z$ to each $Z \in \gamma$. Since there are equivalent sequential polymers being distinguished in the sum over permutations, i.e., permutations of repeated hyperedges, we introduce a factor of $1/\prod_{Z^l}m(Z^l)!$ to avoid overcounting.
\end{proof}
\end{widetext}
With the expression for $w_{\gamma}^t$, we can now study convergence. We proceed as before, first selecting a site $i$ and summing over all polymers connected to it. We choose again the same function $a(\gamma) =\vert\gamma\vert$, and we proceed as in Appendix~\ref{app:KP_LR} selecting an arbitrary hyperedge $Z^*$ 
\begin{align}
    \sum_{ \gamma \nsim Z^* } \vert w_\gamma^t\vert \mathrm{e}^{a\left(\gamma\right)} &\leq \sum_{\gamma \nsim Z^* }\vert w_\gamma^t\vert \mathrm{e}^{\norm{\gamma}} \\
    & \leq \sum_{m=1} \mathrm{e}^{m} \sum_{\substack{t_1,...,t_K \\ \sum_l t_l=m}} \binom{m}{t_1, \dots, t_l, \dots, t_K} \\
    & \cdot \prod_l^{K} \vert\lambda_l\vert^{t_l}  \sum_{\substack{\gamma \nsim Z^* \, \norm{\gamma}=m \\ \text{fixed } t_l  \forall l} }   \prod_{Z \in \gamma} \norm{h_Z}   .
\end{align}
Since all the $\{ h_Z\}$ obey Eq.~\eqref{eq:HkLocal}, irrespective of their color, this is exactly the same sum as in Appendix~\ref{app:KP_LR}, but with the difference that we also have the sum over the combinations of $\{ t_l \}$, so that
\begin{align}
    \sum_{ \gamma \nsim Z^* } \vert w_\gamma^t\vert \mathrm{e}^{a\left(\gamma\right)} & \leq \sum_{m=1} \left(4\mathrm{e}guk\right)^m \sum_{\substack{t_1,...,t_K \\ \sum_l t_l=m}}\\
    &\cdot \prod_l^{K} \vert\lambda_l\vert^{t_l}\binom{m}{t_1, \dots, t_l, \dots, t_K}    \\
    & \leq \sum_{m=1} \left(4 \mathrm{e}guk\right)^m    \left(\sum_l^K \vert \lambda_l \vert\right)^m.
\end{align}
Therefore, by repeating the argument of Appendix~\ref{app:KP_LR}, we conclude that for $\sum_l \vert \lambda_l \vert \leq \beta^* $, Lemma~\ref{le::Kotecky-Preiss} is satisfied, and
\begin{equation}
    \sum_{\Gamma \in \mathcal{G}_C } \vert\varphi \left(\Gamma\right) \prod_{\gamma \in \Gamma} w_{\gamma}^t\vert \leq N.
    \label{eq::IndepBoundGeneralizedClusterExp}
\end{equation}

We again repeat the same process as above, but the quantity to bound now is 
\begin{equation}
    \left \vert \log Z_{\rho, Gen} -T_{M_1,  \dots M_K} \right \vert
\end{equation}
where $T_{M_1,  \dots M_K}$ is the order $\{M_l\}$ expansion. Let $M= \sum_l M_l$, where $M_l$ is the order of the parameter $\lambda_l$ in the expansion. Then
\begin{align}
    \vert \log Z_{\rho, Gen} - T_M \vert &= \vert\sum_{\substack{\Gamma \in \mathcal{G}_C \\ \sum_l n_l \geq M } } \varphi \left(\Gamma\right) \prod_{\gamma \in \Gamma} w_{\gamma}^t\vert \\
    &\leq \sum_{\substack{\Gamma \in \mathcal{G}_C \\ \sum_l n_l \geq M } } \vert \varphi \left(\Gamma\right) \prod_{\gamma \in \Gamma} w_{\gamma}^t\vert \\
    &= \sum_{\substack{\{n_l\}\\ \sum_l n_l \geq M }}^\infty \sum_{\substack{\Gamma \in \mathcal{G}_C \\ \{n_l \} } } \vert \varphi \left(\Gamma\right) \prod_{\gamma \in \Gamma} w_{\gamma}^t\vert,
\end{align}
with $n_l$ being the number of hyperedges of the color $l$ of a certain cluster, so in the last step we simply rearrange the sum in terms of number of hyperedges of each color. Notice that once we fix $n_l$ we can write

\begin{align}
    \vert &\log Z_{\rho, Gen}-T_{M_1,  \dots M_K} \vert \\
    &\leq \sum_{\substack{\{n_l\}\\\sum_l n_l \geq M}}^\infty \left( \prod_{n_l'}\bigg \vert \frac{\lambda_l}{\beta^*/K} \bigg \vert^{n_l'} \right) \\
    &\cdot \sum_{\substack{\Gamma \in \mathcal{G}_C \\ \{n_l \} } }  \vert \varphi   \left(\Gamma\right) \prod_{\gamma \in \Gamma} w_{\gamma}^t\left[ \lambda_l = \beta^*/K\right]\vert,
\end{align}
where $n_l' \leq n_l$ and we chose them such that $\sum_l n_l'=M$. Upper bounding by simply summing over this new variable we obtain
\begin{align}
    \vert &\log Z_{\rho, Gen} - T_{M_1,  \dots M_K} \vert \\
    &\leq \sum_{\{n_l\}\geq \{M_l\}}^\infty \sum_{\substack{n_l' \\\sum_l n_l' =M}} \left( \prod_{n_l'}\bigg \vert \frac{\lambda_l}{\beta^*/K} \bigg \vert^{n_l'} \right) \times\\
    & \sum_{\substack{\Gamma \in \mathcal{G}_C \\ \{n_l \} } }  \vert \varphi   \left(\Gamma\right) \prod_{\gamma \in \Gamma} w_{\gamma}^t\left[ \lambda_l = \beta^*/K\right]\vert \\
    &\leq \left( \sum_{l} \bigg \vert \frac{\lambda_l}{\beta^*/K} \bigg \vert \right)^M \sum_{\{n_l\}\geq \{M_l\}}^\infty \times\\
    &\sum_{\substack{\Gamma \in \mathcal{G}_C \\ \{n_l \} } }  \vert \varphi   \left(\Gamma\right) \prod_{\gamma \in \Gamma} w_{\gamma}^t\left[ \lambda_l = \beta^*/K\right]\vert \\
    &\leq \left(K \sum_{l} \bigg \vert \frac{\lambda_l}{\beta^*} \bigg \vert \right)^M N,
\end{align}
where in the last inequality we used Eq.~\eqref{eq::IndepBoundGeneralizedClusterExp}. Thereby the cluster expansion of $\log Z_{\rho, Gen}$ converges for $\sum_l \vert \lambda_l \vert \leq \beta^*/K$.
\section{Proofs of statistical properties}
If we expand $\log \langle \mathrm{e}^{\tau A} \rangle_\rho$ in $\tau$ we obtain
\begin{equation}
    \log \langle \mathrm{e}^{\tau A} \rangle_\rho =  \langle A \rangle_\rho \tau +\frac{1}{2} \left(\langle A^2 \rangle_\rho-\langle A \rangle_\rho ^2\right) \tau^2 + \mathcal{O}(\tau^3),
    \label{eq:ap_expansionwrtTau}
\end{equation}
where the remainder term can be controlled by Theorems~\ref{Theorem:KoteckyLR} and~\ref{theo:KP_GenPartFunct}.
At this point is useful to remember that, as pointed out by Dobrushin~\cite{Dobrushin1996, AlgorithmicPirogov–Sinai}, the cluster expansion and Taylor series are the same series in $\tau$ but arranged differently, so we can unequivocally identify the moments of both series.
\subsection{Chernoff-Hoeffding proof}\label{app:CHbound}
If we set $\rho = \rho_\beta$ and $\langle A \rangle_\beta=0$ all the terms that contribute to the expansion given by Eq.~\eqref{eq:ap_expansionwrtTau} will have at least a common factor $\tau^2$. In a similar manner to Appendix~\ref{app:KP_LR}, and using that $\tau +\beta \leq \beta^*/2$, we can write
\begin{align}
    \log &\langle \mathrm{e}^{\tau A} \rangle_\beta \\
    &= \sum_{\substack{\Gamma \in \mathcal{G}_C \\ n_\beta \geq 0 \, ;  \, n_\tau \geq 1 } } \varphi \left(\Gamma\right) \prod_{\gamma \in \Gamma} w_{\gamma}^t\\
    &=\left(\frac{\tau}{\beta^*/2-\beta} \right)^2 \times \\
    &\sum_{\substack{\Gamma \in \mathcal{G}_C \\ n_\beta \geq 0 \, ;  \, n_\tau \geq 1 } } \varphi \left(\Gamma\right) \prod_{\gamma \in \Gamma} w_{\gamma}^t \left[\tau = \beta^*/2-\beta\right] \\
    &\leq c_\beta N \tau^2, 
\end{align}
with $c_\beta= \left(\frac{1}{\beta^*/2-\beta}\right)^2$, and where we used Eq.~\eqref{eq::IndepBoundGeneralizedClusterExp} in the last step. Therefore,
\begin{align}
&P_{A,\beta} (| a - \langle A \rangle_{\beta} | > \delta) \\
&= \int_{a- \langle A \rangle_\beta\geq \delta} \Tr \left[ \rho \delta \left(a-A\right)\right] \\
&=\int_{a- \langle A \rangle_\beta\geq \delta} \Tr \left[ \rho \mathrm{e}^{\tau \left(A- \langle A \rangle_\beta\right)}\mathrm{e}^{-\tau \left(A- \langle A \rangle_\beta\right)} \delta \left(a-A\right)\right] \\ 
&\leq \mathrm{e}^{-\tau \delta } \Tr \left[\rho \mathrm{e}^{\tau \left(A- \langle A \rangle_\beta\right)}\right]  \leq \mathrm{e}^{-\tau \delta } \mathrm{e}^{ c_\beta N \tau^2 }
\end{align}

If we set $\tau = \frac{\delta}{2c_\beta N}$ in the last expression we arrive to the bound
\begin{equation}
    P_{A,\beta} (|a - \langle A \rangle_{\beta} | > \delta) \leq \exp \left[- \frac{\delta^2}{4c_\beta N}\right].
\end{equation}
This choice of $\tau$ implies the constrain $\delta \leq  2N (\beta^*/2-\beta)^{-1} \leq 8gukN$, but it turns out to be trivial since by definition $\delta \leq 2\norm{A}\leq 2\sum_{i,j} \norm{a_{ij}} \leq 2ugN$. Therefore the bound holds for $\beta\leq \beta^*/2$. For a product state the proof is identical but without the dependence in $\beta$ of $c_\beta$, being instead just $c_\rho = (\beta^*)^{-2}$~\cite{classsimofshort_PRXQuantum.4.020340}.
\subsection{Berry-Esseen error proof}\label{app:BE}
In order to prove Theorem~\ref{th:LR_BEtheorem} we set $\tau=i\lambda$ in \eqref{eq:ap_expansionwrtTau}. We prove the case of $\rho =\rho_\beta$.
From Eq.~\eqref{eq:ap_expansionwrtTau} we conclude that all the quadratic contributions in time of the expansion sum up to $\frac{-\lambda^2 \sigma_\beta^2}{2}$. If we set $\langle A \rangle_\rho =0 $ all the linear terms in $\lambda$ vanish, so we can write
\begin{equation}
   \left \vert \log \Tr\left[\mathrm{e}^{i\lambda A} \rho_\beta  \right] - \frac{-\lambda^2 \sigma_\beta^2}{2} \right \vert = \vert \sum_{\substack{\Gamma \in \mathcal{G}_C \\n_\lambda \geq 3} }\varphi \left(\Gamma\right) \prod_{\gamma \in \Gamma} w_{\gamma}^t \vert,
\end{equation}
where $n_\lambda$ is the number of edges of the variable $\lambda$ in the cluster. Now in a way similar to Appendix~\ref{app:CHbound},
\begin{align}
  & \big \vert \sum_{\substack{\Gamma \in \mathcal{G}_C \\n_\lambda \geq 3} }\varphi \left(\Gamma\right) \prod_{\gamma \in \Gamma} w_{\gamma}^t  \big\vert \\
  &\leq \left\vert\frac{\lambda}{\beta^*/2-\beta}\right\vert^3 \sum_{\substack{\Gamma \in \mathcal{G}_C } }\big\vert\varphi \left(\Gamma\right) \prod_{\gamma \in \Gamma} w_{\gamma}^t\left[\lambda=\beta^*/2-\beta\right] \big \vert \\
  &\leq N \left\vert\frac{\lambda}{\beta^*/2-\beta}\right\vert^3 
\end{align}
which, for $\lambda \leq \beta^*/2-\beta$, leads to 
\begin{equation}
  \left  \vert \log \Tr\left[\mathrm{e}^{i \lambda A} \rho_\beta \right]- \frac{-\lambda^2 \sigma_\beta^2}{2} \right \vert \leq N\left\vert\frac{\lambda}{\beta^*/2-\beta}\right\vert^3.
\end{equation}
Recall that $\varphi_A(\lambda) =  \langle \mathrm{e}^{-\mathrm{i}\lambda A/\sigma_\beta} \rangle_\beta$ and define $\lambda'=\lambda \sigma_\beta$ so
\begin{equation}
   \left \vert \log \varphi_A (\lambda')- \frac{-\lambda'^2}{2} \right \vert \leq N \left\vert\frac{\lambda}{\left(\beta^*/2-\beta\right) \sigma_\beta}\right\vert^3 .
\end{equation}
Following the derivation of Appendix A in~\cite{Rai_2024} we arrive to
\begin{equation}
    \zeta_N \leq \frac{2C}{\sigma_\beta \left(\beta^*/2-\beta\right)} + \frac{3N}{\left(\sigma_\beta \left(\beta^*/2-\beta\right)\right)^3} \leq \mathcal{O}(N^{-1/2}),
    \end{equation}
for $\lambda < \beta^* \sigma_\beta /2 \equiv \mathcal{O}(N^{1/2})$ and $C\leq \frac{18}{\sqrt{2}\pi}$. For the last inequality, we assume the variance of the Gibbs state $\rho_\beta$ grows as $\Omega (N^{1/2})$. Once again, we proved the result for the Gibbs state, and for a product state the proof is identical but without dependence in $\beta$ (see Appendix A of~\cite{Rai_2024}).

%%%%%
%%%%%

\section{Simulation details}\label{app: sim}

The numerical simulations consist of a TEBD-like algorithm where we construct the time evolution operator of the LR-TFI model~\eqref{eq:def-TFI} explicitly in the short-range case~\cite{MPOTEv_CiracPirvu_2010, MPOTEv_Paeckel_2019} and then applying a superextensive amount of swap gates to exactly simulate the power-law interactions in a finite system size.
This yields a Trotterized representation of the time evolution operator $U(\delta t) = \mathrm{e}^{-\mathrm{i}H \delta t}$ in form of an MPO. For this, we have used the ITensor library in Julia~\cite{ITensor_Fishman_2022}.
We concatenate the Trotterized time evolution to yield the desired final time $ U(t) \simeq \prod_{i=1}^{n_t} U(\delta t)$, with $t = n_t \cdot \delta t$. Once we have the time evolution of the trace $\Tr[U(t)] = \Tr[\mathrm{e}^{-\mathrm{i} H t}]$, a Fast Fourier Transform $\mathcal{F}$ is applied in order to obtain the DOS $P(E)$ as,
\begin{equation}
    P(E) =  \frac{1}{d}\sum_i \delta (E-E_i) =\frac{1}{d} \mathcal{F} \left( \Tr [\mathrm{e}^{-\mathrm{i}Ht}]\right) \,.
\end{equation}
In the simulation, the relevant parameters are the time-step $\delta t$ and the SVD cutoff, which have been determined in order not to saturate a maximal bond dimension of $\chi_\mathrm{max} = 1800$ during the time evolution.
To assess the accuracy of our simulations we have computed the CDOS obtained with our method against the one obtained via exact diagonalization in Fig.~\ref{fig:methodology_check}.

\begin{figure*}[h]
    \centering
    \includegraphics[width=0.65\linewidth]{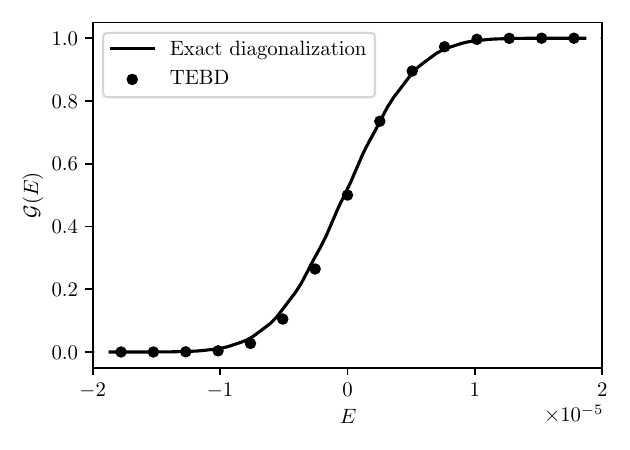}
    \caption{\justifying
    Cumulative density of States (CDOS) estimated by our TN simulation (dots) and via exact diagonalization (line) for the transverse-field Ising model \eqref{eq:def-TFI} in system size $N=20$. The parameters of the simulation are $\delta t=0.005$, $h/J=0.25$, and $\text{SVD cutoff}=10^{-22}$ with a maximum bond dimension of $\chi_\mathrm{max} = 1800$.
    }
    \label{fig:methodology_check}
\end{figure*}

\end{document}